\documentclass[12pt]{article}

\usepackage[margin=1in]{geometry}  
\usepackage{graphicx}
\usepackage{mathrsfs}

\usepackage{amssymb}
 \usepackage{amsthm}
 \usepackage{amsmath}
 \usepackage{amsfonts}  

\newcommand\abs[1]{\left|#1\right|}
\numberwithin{equation}{section}
\usepackage[numbers,sort&compress]{natbib}    

\newtheorem{thm}{Theorem}[section]
\newtheorem{lem}[thm]{Lemma}

\begin{document}
\nocite{*}

\title{Infinite energy solutions to vortex equations governing the fractional quantum Hall effect}

 \author{Joseph Esposito\footnote{Email address: esposito@cims.nyu.edu}\\
Courant Institute of Mathematical Sciences\\
New York University\\
New York, NY 10012}
\maketitle
\providecommand{\keywords}[1]{{{Keywords:}} #1}
\providecommand{\MSC}[1]{{{MSC numbers:}} #1}

\begin{abstract}
In this paper, we utilize weighted Sobolev spaces to establish an existence theory for infinite energy solutions to a coupled non-linear elliptic system.  This system describes the fractional quantum Hall effect in two-dimensional double-layered systems. Via variational methods in a suitable weighted Sobolev space, we prove the existence of multiple vortices over the full plane.  These methods include constrained minimization of an action functional and the existence of a critical point, known as a saddle point, by way of a mountain pass theorem.  Furthermore, for these solutions, which are necessarily of infinite energy, we establish exponential decay estimates.
\end{abstract}
 \smallskip
 \\\\
%
%
\MSC{35J20, 81V70, 46E35, 46N50}

\section{Introduction}
\label{introduction}
In this paper, we study the nonlinear elliptic system over $\mathbb{R}^2$,
\begin{equation}\label{vortex0}
\begin{cases}
\Delta u=4k_{11}e^u+4k_{12}e^v-4+4\pi\displaystyle\sum\limits_{j=1}^{N_1}\delta_{p_j}(x)\\
\Delta v=4k_{21}e^u+4k_{22}e^v-4+4\pi\displaystyle\sum\limits_{j=1}^{N_2}\delta_{q_j}(x)
\end{cases}
\end{equation}
where $\delta_{p_j}$ and $\delta_{q_j}$ are the Dirac distributions centered at $p_j$ and $q_j$ respectively.  We also note that the symmetric coupling matrix, $K=(k_{ij})$, is given by
\begin{equation}\label{kmat}
\nonumber K=\frac{1}{p}
\begin{pmatrix}
p+q&p-q\\p-q&p+q
\end{pmatrix},
\end{equation}
with $p,q\in \mathbb{R}$, $p>0$, $q\neq0$.  In this paper, we are interested only in a coupled system.  Therefore, we require that $p\neq q$.  

The goal of this paper is to study the coupled non-linear elliptic system (\ref{vortex0}) describing the fractional quantum Hall effect in the full space, $\mathbb{R}^2$, when the solutions are not necessarily of finite energy.  As we will see, the solutions obtained in this paper do not satisfy the finite energy conditions established in \cite{medina}, however they will indeed solve the system of differential equations.  In this scenario, the typical function space for bounded domains will not suffice.  That is, we will need to look outside of the standard $L^2$, or $W^{1,2}$, spaces for suitable solutions.  To this end, we will choose a suitable weight for $L^2$, or more appropriately $W^{1,2}$, so that it contains all solutions bounded over the full plane \cite{weighted1,weighted2,weighted3,mcowen1984equation,sy2}.  As a result, we will be able to utilize regularity and embeddings similar to those of standard Sobolev spaces \cite{sy2,mcowen1984equation}.  

We arrive at our results by first taking a series of transformations, along with a change of variables \cite{sy2}, to reduce our elliptic system into triangular form.  Then, we are able to establish a variational principle and via a constrained minimization, we establish existence of non-topological solutions.  We also provide information about the asymptotic behavior of the solutions.

The model describing the FQHE in a two-dimensional double-layered system was developed in terms of Chern-Simons theory \cite{hans,hany,lieby,lin,nam,sy1,sy2,sy3,bt,taran,wangy,yangy1,yangy2,yangy3,yang2013solitons}.  Ichinose and Sekiguci \cite{ichi2} studied the FQHE in double-layer electron systems in terms of Chern-Simons gauge theory of bosonized electrons.  In \cite{ichi2} a radially symmetric ansatz led to an effective theory for topological excitations in the $(m,m,n)$ Halperin state.  This theory was limited to single soliton configurations.  The BPS equations, which we will present in section 2, were developed by Medina \cite{medina} and an existence theory for both doubly-periodic and full-plane domains was developed. The existence theory in \cite{medina} was restricted to only positive definite coupling matrices and topological solutions which were of finite energy. 

	
The structure of this paper is as follows.  In section \ref{BPS}, we present the BPS and vortex equations developed by Ichinose, Sekiguchi, and Medina \cite{ichi2,medina}.  We also discuss the properties of the coupling matrix.  Lastly, we present the appropriate weighted Sobolev space and corresponding analytical tools developed by McOwen \cite{mcowen1984equation}.  In section \ref{weight}, we reduce the system to a triangular one via a change of variables.  In sections \ref{mink12<0} and \ref{min0<k12<1},  we establish existence of a unique solution to a constrained minimization problem when the matrix $K$ satisfies $k_{12}<0$ and $k_{12}>0$ respectively.  In section \ref{mpass}, we establish a mountain pass structure and existence of a saddle point type solutions for a modified functional. Finally, in section \ref{decay}, we study the decay estimates for solutions established in sections \ref{mink12<0}-\ref{mpass} and the discussion follows in section \ref{discuss}.

\section{BPS and Vortex Equations}\label{BPS}
Here, we provide a brief overview of the system of equations governing the FQHE.   The Lagrangian describing the FQHE in 2 dimensional electron systems was well discussed in \cite{ichi2,medina} and is given by
\begin{equation}
\nonumber\mathcal{L}=\mathcal{L}_{\phi}+\mathcal{L}_{CS},
\end{equation}
which is a sum of the matter term,
\begin{equation}
\nonumber\mathcal{L}_{\phi}=i\bar{\psi}_{\uparrow}(\partial_0-i a_0^+-ia_0^-)\psi_{\uparrow}+i\bar{\psi}_{\downarrow}(\partial_0-ia_0^++ia_0^-)\psi_{\downarrow}
-\frac{1}{2M}\sum_{\sigma=\uparrow,\downarrow}\abs{D_j^{\sigma}\psi_{\sigma}}^2-V(\psi_{\uparrow},\psi_{\downarrow}),
\end{equation}
and the Chern-Simons term,
\begin{equation}
\nonumber\mathcal{L}_{CS}=\mathcal{L}_{CS}(a_{\mu}^+)+\mathcal{L}_{CS}(a_{\mu}^-)=-\frac{1}{4}\epsilon_{\mu\nu\lambda}\left(\frac{1}{p}a_{\mu}^+\partial_{\nu}a_{\lambda}^++\frac{1}{q}a_{\mu}^-\partial_{\nu}a_{\lambda}^-\right).
\end{equation}


The bosonized fields are represented in terms of the upper layer, $\psi_{\uparrow}$, and the lower layer, $\psi_{\downarrow}$.  The mass of the electrons is given by $M$, while the parameters $p$ and $q$ are real numbers.  We also have $a_{\mu}^+$ and $a_{\mu}^-$, the scalar potential fields corresponding to $U(1)\otimes U(1)$ symmetry. In \cite{ichi2}, these scalar fields were restricted to a certain behavior which yielded necessary conditions for the total energy of the system to be finite.  In the present, we require no such constraints on the fields and allow for solutions which are of divergent energy.


The first integral of the system of BPS type was obtained by Medina in \cite{medina}.  It is given by
\begin{align}\label{bps}
(D_1^{\uparrow}&-iD_2^{\uparrow})\psi_{\uparrow}=0\\\label{d1}
(D_1^{\downarrow}&-iD_2^{\downarrow})\psi_{\downarrow}=0\\\label{d2}
B_{12}&=2(p+q)\abs{\psi_{\uparrow}}^2+2(p-q)\abs{\psi_{\downarrow}}^2-eB\\
\tilde{B}_{12}&=2(p-q)\abs{\psi_{\uparrow}}^2+2(p+q)\abs{\psi_{\downarrow}}^2-eB\\
b_0&=\frac{1}{M}(p+q)\abs{\psi_{\uparrow}}^2+\frac{1}{M}(p-q)\abs{\psi_{\downarrow}}^2+\frac{eB}{M}\\\label{d6}
\tilde{b}_0&=\frac{1}{M}(p-q)\abs{\psi_{\uparrow}}^2+\frac{1}{M}(p+q)\abs{\psi_{\downarrow}}^2+\frac{eB}{M},
\end{align}
where equations (\ref{bps}) and (\ref{d1}) are the self-dual equations we are looking for the solutions to.  There are many solutions to these self-dual equations, but this degeneracy is removed by the Chern-Simons constraints \cite{ichi2}. The ground state configurations for the fractional quantum Hall effect are represented in terms of the average electron density, $\bar{\rho}$,
\begin{align}
\nonumber\psi_{\uparrow,0}=\psi_{\downarrow,0}=\sqrt{\frac{\bar{\rho}}{2}}
\end{align}
which yields the following integral representing the total energy of the system,
\begin{equation}\label{energy}
E=\frac{1}{2M}\int\displaylimits_{\mathbb{R}^2} \Bigg[ \sum \abs{(D_1^{\sigma}-iD_2^{\sigma})\psi_{\sigma}}^2+eB\left(\abs{\psi_{\uparrow}}^2+\abs{\psi_{\downarrow}}^2\right)-eB\left(\abs{\psi_{\uparrow,0}}^2+\abs{\psi_{\downarrow,0}}^2\right)\Bigg]dx.
\end{equation}
In \cite{ichi2}, radially symmetric solutions over the full plane were established numerically.  In \cite{medina}, solutions to (\ref{bps})--(\ref{d6}) were considered over the full plane under the condition that they were of finite energy.  As a result, the following topological boundary conditions were imposed on the complex fields, $\psi_{\uparrow}$ and $\psi_{\downarrow}$,
\begin{equation}
\abs{\psi_{\uparrow}}^2\rightarrow\abs{\psi_{\uparrow,0}}^2=\frac{\bar{\rho}}{2}\qquad\text{ and }\qquad\abs{\psi_{\downarrow}}^2\rightarrow\abs{\psi_{\downarrow,0}}^2=\frac{\bar{\rho}}{2}\qquad\text{ as }\abs{x}\rightarrow\infty.
\end{equation}
In the present, we are also interested in solutions of the BPS (\ref{bps})--(\ref{d6}) over the full plane.  However, we remove this boundary condition and will see that without it, the total energy of the system (\ref{energy}) is not necessarily finite.  Therefore, we allow for solutions to the system that are of divergent energy and establish the existence of non-topological solutions.

We now seek these divergent energy solutions over the full plane.  To this end,  we identify $\mathbb{R}^2$ with the complex plane, $\mathbb{C}$, let $z$ be a point in $\mathbb{C}$ and let $z_0$ be a zero of $\psi_{\uparrow}$.  The first BPS equation tells us that in a neighborhood of $z=z_0$, with $z=x_1+ix_2$,
\begin{equation*}
\psi_{\uparrow}(z_0)=(z-z_0)^{n_0}\hat{h}_0(x_1,x_2)
\end{equation*} 
where $\hat{h}_0$ is nonzero at $z_0$ and smooth \cite{griffiths2014principles}.  We observe that the zeros of the fields $\psi_{\uparrow}$ and $\psi_{\downarrow}$ are therefore discrete and each forms a finite set represented by $Z_{\psi_{\uparrow}}=\left\{p_1,p_2,\ldots,p_{N_1}\right\}$ and  $Z_{\psi_{\downarrow}}=\left\{q_1,q_2,\ldots,q_{N_2}\right\}$ where the multiplicities of the zeros $z=p_j$ and $z=q_j$ are $n_p^j$ and $n_q^j$ respectively.  We then define
\begin{equation}
u=\ln\abs{\psi_{\uparrow}}^2-\ln\abs{\bar{\rho}}\qquad\text{ and }\qquad v=\ln\abs{\psi_{\downarrow}}^2-\ln\abs{\bar{\rho}},
\end{equation}
and pair this with the transformation $x\mapsto \sqrt{p\bar{\rho}}x$ to yield the vortex equations
\begin{equation}\label{vortex1}
\begin{cases}
\Delta u=4k_{11}e^u+4k_{12}e^v-4+4\pi\displaystyle\sum\limits_{j=1}^{N_1}n_p^j\delta_{p_j}(x)\\
\Delta v=4k_{21}e^u+4k_{22}e^v-4+4\pi\displaystyle\sum\limits_{j=1}^{N_2}n_q^j\delta_{q_j}(x)
\end{cases}
\end{equation}
where $\delta_{p_j}$ and $\delta_{q_j}$ are the Dirac distributions centered at $p_j$ and $q_j$ respectively.  In the rest of this paper, without loss of generality, we allow repetition of $p_j$'s and $q_j$'s and we will treat the $n_p^j$ and $n_q^j$ terms as having value 1. 

Here, we state the main theorem of this paper
\begin{thm}\label{mainthm}
Let $\left\{p_1,\ldots,p_{N_1},q_1,\ldots,q_{N_2}\right\}\subset\mathbb{R}^2.$  For any 
\begin{align*}
\alpha&=\alpha_0+\frac{4k_{12}}{k_{11}}N_1-4N_2>0\\
\beta&=\beta_0-\frac{4k_{12}}{k_{11}}N_1>0
\end{align*}
satisfying
\begin{equation}
0<\beta<\frac{\alpha}{4}\left(\frac{p}{q}+\frac{q}{p}-2\right)\text{ when }q>p>0,
\end{equation}
or

\begin{equation}
\alpha>0\qquad\text{and}\qquad \beta>\frac{\alpha}{4}\left(\frac{p}{q}+\frac{q}{p}-2\right)\text{ when }p>q>0, 
\end{equation}
the system given by (\ref{bps})--(\ref{d6}) has a solution $\left(\psi_{\uparrow}^{(\alpha,\beta)},\psi_{\downarrow}^{(\alpha,\beta)},b_{\mu}^{(\alpha,\beta)},\tilde{b}_{\mu}^{(\alpha,\beta)},b_{0}^{(\alpha,\beta)},\tilde{b}_{0}^{(\alpha,\beta)}\right)$ of locally square integrable functions $\psi_{\uparrow}$ and $\psi_{\downarrow}$.  
\end{thm}

In what follows, we use the notation $\det(K)=\abs{K}=4{q}/{p}$ and note that the eigenvalues of $K$ are $\lambda=2, 2{q}/{p}$.  Since $p\neq q$, we observe that $\abs{K}\neq 4$ which splits the positive definite possibilities into two regions, $0<\abs{K}<4$ and $\abs{K}>4$.  Moreover, it is easy to see that if $p, q$ are of the same sign (both positive), then $K$ is positive definite.  Otherwise, $K$ is indefinite.   In the present study of non-topological solutions, we consider only the positive definite case. 




Here, we define a suitable weighted Sobolev space for the problem that follows and, as in \cite{sy2,mcowen1984equation}, we let $d\mu=h_0dx$ where $h_0\in C^{\infty}$ with
\begin{equation}
\nonumber h_0(x)=\abs{x}^{-\kappa},\qquad\text{ for }\abs{x}\geq 1,\kappa>4.
\end{equation}

We will use the notation $L^p(d\mu)=L^p(\mathbb{R}^2,d\mu)$, and let $\mathscr{H}$ denote the Hilbert space of $L_{\text{loc}}^2$ functions for which the following norm is finite
\begin{equation}
\|u\|_{\mathscr{H}}^2=\|\nabla u\|_{L^2(dx)}^2+\|u\|_{L^2(d\mu)}^2.
\end{equation}
We denote by $\tilde{\mathscr{H}}$ the closed subspace of $\mathscr{H}$,
\begin{equation}
 \tilde{\mathscr{H}}=\left\{u\in\mathscr{H}:\int\displaylimits_{\mathbb{R}^2}ud\mu=0\right\}.
\end{equation}
Therefore, for any $u\in\mathscr{H}$, we can decompose $u$ into
\begin{equation}
u=\bar{u}+u',\qquad \bar{u}\in\mathbb{R},\qquad u'\in\tilde{\mathscr{H}}.
\end{equation}

Before proceeding, we state some necessary lemmas, proofs of which can be found in \cite{sy2,mcowen1984equation}.  They will play an integral role in our study of the solutions to the system given by (\ref{vortex1}).
\begin{lem}\label{tm}
For any $0<\varepsilon<4\pi$, there is a $C(\varepsilon)>0$ so that
\begin{equation*}
\int\displaylimits_{\mathbb{R}^2}e^{a\abs{u}}d\mu\leq C(\varepsilon)e^{\frac{a^2}{4(4\pi-\varepsilon)}\| \nabla u\|_{L^2(dx)}^2},\qquad u\in\tilde{\mathscr{H}}
\end{equation*}
for any $a\in\mathbb{R}$.
\end{lem}

\begin{lem}\label{poin}
The Poincar\'e inequality holds on $\mathscr{H}$.  In other words, there is a constant $C>0$ such that
\begin{equation*}
\|u\|_{L^2(d\mu)}^2\leq C\| \nabla u\|_{L^2(dx)}^2,\qquad u\in\tilde{\mathscr{H}}.
\end{equation*}
\end{lem}
Finally, we have
\begin{lem}\label{embed}
The injection $\mathscr{H}\rightarrow L^2(d\mu)$ is a compact embedding.
\end{lem}

Now we are ready to modify the system of equations (\ref{vortex1}) accordingly.  In the next section, we first transform the system into triangular form and then establish a variational principle.
\section{System in view of weighted Sobolev spaces}\label{weight}

Here, we consider the two dimensional system (\ref{vortex1}) over $\mathbb{R}^2$ and rewrite our system so that we can establish a variational principle.  To deal with the delta functions, we must introduce a cutoff function.  We follow the procedure of Yang in \cite{yang2013solitons}.  To this end, we define $\rho(t)$ to be a smooth increasing function over $t>0$ so that
\begin{equation}\label{rhot}
\rho(t)=\begin{cases}\ln t,& t \leq \frac{1}{2},\\ 
0,& t \geq 1,\\ 
\leq 0,& \text{ for all }t>0.  
\end{cases}
\end{equation}
We now cover the points $p_j$ with disjoint open balls. Let $\delta>0$ be so that
\begin{equation}
B_{\delta_{p_j}}=\left\{x \Big|\abs{x-p_j}<\delta\right\},\qquad j=1,2,\ldots,N_1,
\end{equation}
and consider the functions
\begin{equation}
\phi_j(x)=2\rho\left(\frac{\abs{x-p_j}}{\delta}\right),\qquad j=1,2,\ldots,N_1.
\end{equation}
We also consider the functions with compact support given by
\begin{equation}
\phi_{0,j}=\Delta \phi_j-4\pi\delta_{p_j}(x).
\end{equation}
It is clear that
\begin{equation}
\int\displaylimits_{\mathbb{R}^2}\phi_{0,j}dx=4\pi,
\end{equation}
with
\begin{align}
\phi_j(x)\leq 0\qquad\text{ and }\qquad
\phi_j(x)=0\qquad\text{ when }\qquad\abs{x-p_j}\geq\delta.
\end{align}
We now define a subtractive background function, $\displaystyle u_0=\sum_{j=1}^{N_1}\phi_j$, and observe that
\begin{equation}\label{subg}
\Delta u_0=4\pi\sum_{j=1}^{N_1}\delta_{p_j}-g_1.
\end{equation}
The function $g_1$ in (\ref{subg}) satisfies
\begin{equation}
g_1=\sum_{j=1}^{N_1}\phi_{0,j}\qquad\text{and}\qquad\int\displaylimits_{\mathbb{R}^2}g_1dx=4\pi N_1.
\end{equation}
We now have constructed a function, $u_0$, that satisfies
\begin{equation}
u_0\leq 0\qquad\text{and}\qquad u_0=0,\text{ on }\mathbb{R}^2\setminus\bigcup_{j=1}^{N_1}B_{\delta}(p_j).
\end{equation}
Similarly, we find a cover of $\left\{q_1,q_2,\ldots,q_{N_2}\right\}$ and  define a background function $v_0$ along with a function $g_2$ so that
\begin{equation}
\Delta v_0=4\pi\sum_{j=1}^{N_2}\delta_{q_j}-g_2\qquad\text{and}\qquad
\int\displaylimits_{\mathbb{R}^2}g_2=4\pi N_2.
\end{equation}
We now introduce the substitutions
\begin{align}
u=u_0+u_1-\abs{x}^2\qquad\text{and}\qquad
v=v_0+v_1-\abs{x}^2.
\end{align}
We are then able to obtain a system in terms of $u_1$ and $v_1$,
\begin{equation}\label{sysuv1}
\begin{cases}
\Delta u_1=e^{-\abs{x}^2}\left(k_{11}U_0e^{u_1}+k_{12}V_0e^{v_1}\right)+g_1\\
\Delta v_1=e^{-\abs{x}^2}\left(k_{12}U_0e^{u_1}+k_{11}V_0e^{v_1}\right)+g_2
\end{cases}
\end{equation}
where $U_0=e^{u_0}$, $V_0=e^{v_0}$ and we have used the symmetry of the coupling matrix $K$ to write $k_{12}=k_{21}$ and $k_{11}=k_{22}$. In order to establish a variational principle, we rewrite the system in triangular form.  To this end, we consider the substitutions
\begin{equation}
\label{sub}
u_2=-\frac{2k_{12}}{k_{11}}u_1+2v_1\qquad\text{and}\qquad
v_2=-\frac{2k_{12}}{k_{11}}u_1.
\end{equation}
In view of (\ref{sub}), we are able to rewrite the system given by (\ref{sysuv1}) as
\begin{equation}
\begin{cases}
\Delta u_2=\frac{2\abs{K}}{k_{11}}e^{-\abs{x}^2} V_0e^{\frac{1}{2}(u_2-v_2)}-\frac{2k_{12}}{k_{11}}g_1+2g_2\\
\Delta v_2=-2k_{12}e^{-\abs{x}^2} U_0 e^{-\frac{k_{11}}{2k_{12}}v_2 }-2\frac{k_{12}^2}{k_{11}}e^{-\abs{x}^2} V_0 e^{\frac{1}{2}(u_2-v_2)}-\frac{2k_{12}}{k_{11}}g_1.
\end{cases}
\end{equation}
We now introduce functions $u_3,v_3\in C^{\infty}(\mathbb{R}^2)$ so that
\begin{align}
\begin{split}\label{u3v3}
u_3&=\alpha_0\ln r,\qquad r\geq 1, \alpha_0 >0\\
v_3&=-\beta_0\ln r,\qquad r\geq 1,\beta_0 >0,
\end{split}
\end{align}
and set $f=-\Delta u_3-\frac{2k_{12}}{k_{11}}g_1+2g_2, h=-\Delta v_3-\frac{2k_{12}}{k_{11}}g_1$.  Since $g_1$ and $g_2$ are functions we defined to have compact support,  we can see that $f$ and $h$ are also functions with compact support. The function $f$ satisfies
\begin{align}
\nonumber\int\displaylimits_{\mathbb{R}^2}fdx&=-\int\displaylimits_{\abs{x}\leq 1} \Delta u_3dx-\int\displaylimits_{\mathbb{R}^2}\left(\frac{2k_{12}}{k_{11}}g_1-2g_2 \right)dx\\
&\nonumber=-\int\displaylimits_{\abs{x}=1}\frac{\partial u_3}{\partial r}ds-\frac{2k_{12}}{k_{11}}4\pi N_1+8\pi N_2\\
&\nonumber=-2\pi\alpha_0-\frac{2k_{12}}{k_{11}}4\pi N_1+8\pi N_2\\
&=-2\pi\alpha
\end{align}
where 
\begin{equation}\label{alpha0}
\alpha=\alpha_0+\frac{4k_{12}}{k_{11}}N_1-4N_2>0.
\end{equation}
Similarly, for $h$ we obtain
\begin{equation}
\int\displaylimits_{\mathbb{R}^2}hdx=2\pi\beta,
\end{equation}
where
\begin{equation}\label{beta0}
\beta=\beta_0-\frac{4k_{12}}{k_{11}}N_1>0.
\end{equation}
In order for (\ref{u3v3}) to remain valid, we see the requirements of $\alpha_0$ and $\beta_0$ are 
\begin{equation*}
\alpha_0>-\frac{4k_{12}}{k_{11}}N_1+4N_2,\qquad\text{ and }\qquad \beta_0>\frac{4k_{12}}{k_{11}}N_1.
\end{equation*}
Now, we make the translations 
\begin{equation}
\nonumber u_2=u_3+\xi\qquad\text{and}\qquad v_2=v_3+\zeta,
\end{equation}
to obtain the following triangular system in terms of $\xi$ and $\zeta$,
\begin{equation}\label{main}
\begin{cases}
\Delta \xi=\frac{2\abs{K}}{k_{11}}Ve^{\frac{1}{2}(\xi-\zeta)}+f\\
\Delta \zeta =-2k_{12}Ue^{-\frac{k_{11}}{2k_{12}}\zeta}-2\frac{k_{12}^2}{k_{11}}Ve^{\frac{1}{2}(\xi-\zeta)}+h.
\end{cases}
\end{equation}
We observe that
\begin{equation}
U=e^{-\abs{x}^2} U_0 e^{-\frac{k_{11}}{2k_{12}}v_3}\qquad\text{and}\qquad
V=e^{-\abs{x}^2} V_0 e^{\frac{1}{2}(u_3-v_3)},
\end{equation}
where both $U$ and $V=O\left(e^{-\abs{x}}\right)$ for large $\abs{x}$.  We now consider the functional
\begin{equation}\label{isig}
I(\xi,\zeta)=\int\displaylimits_{\mathbb{R}^2}\left(\frac{1}{2}\abs{\nabla\xi}^2+\frac{\sigma}{2}\abs{\nabla\zeta}^2+f\xi+\sigma h\zeta\right) dx,
\end{equation}
along with
\begin{equation}
J_{1}(\xi,\zeta)=\int\displaylimits_{\mathbb{R}^2}Ve^{\frac{1}{2}(\xi-\zeta)}dx\qquad\text{and}\qquad
J_{2}(\xi,\zeta)=\int\displaylimits_{\mathbb{R}^2}Ue^{-\frac{k_{11}}{2k_{12}}\zeta}dx,
\end{equation}
where $\sigma$ is a real valued constant whose value we will determine in the work that follows.  Prior to determining the value of $\sigma$, we first follow the ideas of Spruck and Yang in \cite{sy1} by relating the first equation in the system (\ref{main}) to McOwen's study of conformal deformation equations in  \cite{mcowen1984equation} and consider a suitable weighted Sobolev space for the functions $\xi$ and $\zeta$.   The true motivation for such a choice is that our solutions, $\xi$ and $\zeta$, may (and actually will) approach constants as $\abs{x}\rightarrow\infty$.  This results in 
\begin{equation*}
\|\xi\|_{L^2(dx)}\rightarrow \infty\qquad\text{and}\qquad\|\zeta\|_{L^2(dx)}\rightarrow\infty\qquad\text{as}\qquad\abs{x}\rightarrow\infty.
\end{equation*}
Therefore, the energy norm, $\|\cdot\|_{W^{1,2}(\mathbb{R}^2)}^2$ is infinite.  The introduction of the power weight for the $L^2(dx)$ norm allows for such behavior while keeping the energy norm finite.

We will now seek to minimize the functional (\ref{isig}) with respect to the following constraints which come from the integration of (\ref{main}).
\begin{align}\label{constraints}
\begin{split}
\frac{2\abs{K}}{k_{11}}&\int\displaylimits_{\mathbb{R}^2}Ve^{\frac{1}{2}(\xi-\zeta)}dx=2\pi\alpha\\
2k_{12}&\int\displaylimits_{\mathbb{R}^2}Ue^{-\frac{k_{11}}{2k_{12}}\zeta}dx+2\frac{k_{12}^2}{k_{11}}\int\displaylimits_{\mathbb{R}^2}Ve^{\frac{1}{2}(\xi-\zeta)}dx=2\pi\beta.
\end{split}
\end{align}
From both equations in (\ref{constraints}), we obtain
\begin{equation}\label{constraints1}
2k_{12}\int\displaylimits_{\mathbb{R}^2}Ue^{-\frac{k_{11}}{2k_{12}}\zeta}dx=2\pi\left(\beta-\frac{\alpha k_{12}^2}{\abs{K}}\right).
\end{equation}

We need to ensure that the constraints are valid.  To this end, we recall that when $K$ is positive definite, $k_{12}<1$ with $k_{12}\neq 0$ and see that it will be necessary to consider two possible intervals for $k_{12}$ and the conditions imposed on $\alpha$ and $\beta$,
\begin{equation*}
\text{when } 0<k_{12}<1,\qquad \beta>\frac{\alpha k_{12}^2}{\abs{K}}
\end{equation*}
and
\begin{equation*}
\text{when }k_{12}<0, \qquad0<\beta<\frac{\alpha k_{12}^2}{\abs{K}}.
\end{equation*}

In the existence theory that that follows, we will discuss these cases separately.  However, for both of these cases, we will consider the following minimization problem
\begin{equation}\label{min}
\min\left\{I(\xi,\zeta)|\xi,\zeta\in\mathscr{H},(\xi,\zeta) \text{ satisfies the constraints (\ref{constraints})}\right\}.
\end{equation}
We now state the main result of this section, which will use the Lagrange multiplier rule to determine the value of $\sigma$ in (\ref{isig}).
\begin{lem}
If $\sigma=\frac{\abs{K}}{k_{12}^2}$, then a solution $(\xi,\zeta)$ of (\ref{min}) is a solution of (\ref{main}).
\end{lem}

\begin{proof}
Let $(\xi,\zeta)$ be a solution of (\ref{min}).  The Lagrange multiplier rule states that there are constants $\lambda_1,\lambda_2\in\mathbb{R}$ such that

\begin{equation}\label{chi1}
\int\displaylimits_{\mathbb{R}^2}\nabla\xi\cdot\nabla\chi_1+f\chi_1dx=\lambda_1\frac{\abs{K}}{k_{11}}\int\displaylimits_{\mathbb{R}^2}Ve^{\frac{1}{2}(\xi-\zeta)}\chi_1dx-\lambda_2\frac{k_{12}^2}{k_{11}}\int\displaylimits_{\mathbb{R}^2}Ve^{\frac{1}{2}(\xi-\zeta)}\chi_1dx, \quad\forall\chi_1\in\mathscr{H},
\end{equation}
and
\begin{multline}\label{chi2}
\int\displaylimits_{\mathbb{R}^2}\sigma\nabla\zeta\cdot\nabla\chi_2+\sigma h\chi_2dx=-\lambda_1\frac{\abs{K}}{k_{11}}\int\displaylimits_{\mathbb{R}^2}Ve^{\frac{1}{2}(\xi-\zeta)}\chi_2dx\\-
\lambda_2k_{11}\int\displaylimits_{\mathbb{R}^2}Ue^{-\frac{k_{11}}{2k_{12}}\zeta}\chi_2dx+\lambda_2\frac{k_{12}^2}{k_{11}}\int\displaylimits_{\mathbb{R}^2}Ve^{\frac{1}{2}(\xi-\zeta)}\chi_2dx,\qquad\forall\chi_2\in\mathscr{H}.
\end{multline}
We combine terms and use the substitutions
\begin{equation*}
\lambda=\lambda_1-\lambda_2\frac{k_{12}^2}{\abs{K}}\qquad
\mu=-\lambda_2\frac{k_{11}}{2k_{12}}.
\end{equation*}
We proceed in terms of $\lambda$ and $\mu$ to obtain
\begin{equation}
\int\displaylimits_{\mathbb{R}^2}\nabla\xi\cdot\nabla\chi_1+f\chi_1dx=\lambda\frac{\abs{K}}{k_{11}}\int\displaylimits_{\mathbb{R}^2}Ve^{\frac{1}{2}(\xi-\zeta)}\chi_1dx,\label{lagrange1}
\end{equation}
and
\begin{equation}
\int\displaylimits_{\mathbb{R}^2}\left(\sigma\nabla\zeta\cdot\nabla\chi_2+\sigma h\chi_2\right)dx=-\lambda\frac{\abs{K}}{k_{11}}\int\displaylimits_{\mathbb{R}^2}Ve^{\frac{1}{2}(\xi-\zeta)}\chi_2dx-2\mu k_{12}\int\displaylimits_{\mathbb{R}^2}Ue^{-\frac{k_{11}}{2k_{12}}\zeta}\chi_2dx\label{lagrange2}.
\end{equation}
Since $\chi_1$ and $\chi_2$ are arbitrary test functions in $\mathscr{H}$, we set $\chi_1\equiv 1$ in (\ref{lagrange1}) and obtain
\begin{equation}
-2\pi\alpha=\frac{1}{2}\lambda(2\pi\alpha).
\end{equation}
We see that by setting $\lambda=-2$, we are able to recover the first equation in (\ref{main}).  Similarly, we set $\chi_2\equiv 1$ in (\ref{lagrange2}) to obtain
\begin{equation}\label{mu}
2\pi\beta\sigma=2\pi\alpha+\mu\left[2\pi\left(\beta-\frac{\alpha k_{12}^2}{\abs{K}} 
\right)\right],
\end{equation} 
and we obtain the second equation in (\ref{main}) by setting $\mu={\abs{K}}/{k_{12}^2}$.  Moreover, in order for this to be true, it is necessary that 
\begin{equation*}
\sigma=\mu=\frac{\abs{K}}{k_{12}^2}.
\end{equation*}

\end{proof}
From this moment on, we fix $\sigma={\abs{K}}/{k_{12}^2}$ and show that the minimization problem given by (\ref{min}), with this value of $\sigma$, has a solution.  In order to do so, we will need to use the decomposition of $\mathscr{H}=\tilde{\mathscr{H}}\oplus \mathbb{R}$ to write,
\begin{equation*}
\xi=\bar{\xi}+\xi',\qquad \zeta=\bar{\zeta}+\zeta',
\end{equation*}
where $\bar{\xi},\bar{\zeta}\in\mathbb{R}$ are constants and $\xi',\zeta'\in\mathcal{\tilde{H}}$. We now are able to rewrite the functional in the form

\begin{equation}\label{func}
I(\xi,\zeta)=\int\displaylimits_{\mathbb{R}^2}\left(\frac{1}{2}\abs{\nabla\xi'}^2+\frac{\abs{K}}{2k_{12}^2}\abs{\nabla\zeta'}^2\right)dx
+\int\displaylimits_{\mathbb{R}^2}\left(f\xi'+\frac{\abs{K}}{k_{12}^2}h\zeta'\right)dx-2\pi\alpha\bar{\xi}+2\pi\beta\frac{\abs{K}}{k_{12}^2}\bar{\zeta}.
\end{equation}
We now proceed to establish existence of a minimizer to (\ref{min}).  In the next section, we consider the case when $k_{12}<0$.

\section{Constrained minimization with $k_{12}<0$}\label{mink12<0}
In this section, we consider the case when $K$ is positive definite with $k_{12}<0$ and prove that the minimization problem (\ref{min}) has a solution.  
The main theorem of this section is given below.

\begin{thm}\label{k12min}
The constrained minimization problem (\ref{min}) has a solution when
\begin{equation*}
0<\beta<\frac{\alpha}{4}\left(\frac{p}{q}+\frac{q}{p}-2\right),\qquad \alpha>0.
\end{equation*}
\end{thm}

\textbf{Remark:} The above inequality is valid.  To see this we note that when $k_{12}<0$, both $q$ and $p$ are positive.  Furthermore, $0<p<q$, and the last factor on the right hand side satisfies
\begin{equation*}
\frac{p}{q}+\frac{q}{p}-2>0.
\end{equation*}  
We now prove Theorem {\ref{k12min}}
\begin{proof}
To show that the minimization problem has a solution, we will need to show that the functional $I(\xi,\zeta)$ given by (\ref{func}) is bounded below, coercive, and weakly lower semicontinuous. In order to find a minimum for the functional given by (\ref{func}),  we will first find a lower bound for the last two terms.  These are given by
\begin{equation}\label{lam}
\Lambda=-2\pi\alpha\bar{\xi}+2\pi\beta\frac{\abs{K}}{k_{12}^2}\bar{\zeta}.
\end{equation} 
From the constraints given by (\ref{constraints}) and (\ref{constraints1}), we obtain the following expressions for $\bar{\xi}$ and $\bar{\zeta}$,
\begin{equation}
\bar{\xi}=\bar{\zeta}+2\ln\left(\frac{k_{11}\pi\alpha}{\abs{K}}\right)-2\ln\left[\int\displaylimits_{\mathbb{R}^2}Ve^{\frac{1}{2}(\xi'-\zeta')}dx \right],
\delimitershortfall 0pt
\end{equation}
and
\begin{equation}\label{UU}
\bar{\zeta}=\frac{2 k_{12}}{k_{11}} \ln\left( \frac{\pi\beta}{k_{12}} -\frac{\pi\alpha k_{12}}{\abs{K}}\right)+\frac{2k_{12}}{k_{11}}\ln\left[\int\displaylimits_{\mathbb{R}^2}Ue^{-\frac{k_{11}}{2k_{12}}\zeta'} dx\right].
\delimitershortfall 0pt
\end{equation}
It is useful to rewrite (\ref{lam}) in the following form,
\begin{equation}\label{lam2}
\Lambda=-2\pi\alpha\bar{\xi}+2\pi\beta\frac{\abs{K}}{k_{12}^2}\bar{\zeta}=\frac{2\pi \abs{K}}{k_{12}^2}\left(\beta-\frac{\alpha k_{12}^2}{\abs{k}}\right)\bar{\zeta}+4\pi\alpha\ln\left[\int\displaylimits_{\mathbb{R}^2}Ve^{\frac{1}{2}(\xi'-\zeta')}dx\right].
\delimitershortfall 0pt
\end{equation}
Since $4\pi\alpha>0$, we must find a lower bound for the last term in (\ref{lam2}). To this end, we recall that
\begin{equation}
 \nonumber V=e^{-\abs{x}^2} V_0 e^{\frac{1}{2}(u_3-v_3)}=e^{-\abs{x}^2}e^{\frac{1}{2}(u_3-v_3)}\prod_{j=1}^{N_2}\abs{x-q_j}^2
 =e^{v_0+\frac{1}{2}(u_3-v_3)-\abs{x}^2}
\end{equation}
and see that  $v_0, u_3,v_3$ and $\abs{x}^2$ all belong to $L(d\mu)$. This, along with Jensen's inequality, allows us to obtain
\begin{align}\label{lam1}
\int\displaylimits_{\mathbb{R}^2}Ve^{\frac{1}{2}(\xi'-\zeta')}dx&=\int\displaylimits_{\mathbb{R}^2}h_0^{-1}Ve^{\frac{1}{2}(\xi'-\zeta')}d\mu\nonumber\\
&\geq \varepsilon_0\int\displaylimits_{\mathbb{R}^2}Ve^{\frac{1}{2}(\xi'-\zeta')}d\mu\nonumber\\
&\geq C_1\exp\left[\int\displaylimits_{\mathbb{R}^2}\left(v_0+\frac{1}{2}(u_3-v_3)-\abs{x}^2\right)d\mu \delimitershortfall 0pt\middle/ \int\displaylimits_{\mathbb{R}^2}d\mu\right].
\delimitershortfall 0pt
\end{align}
where $\int\displaylimits_{\mathbb{R}^2}d\mu$ is finite.  Since the coefficient of 
$\ln\left[\int\displaylimits_{\mathbb{R}^2}Ue^{-\frac{k_{11}}{2k_{12}}\zeta'} dx\right]$ is negative in (\ref{UU}), we must find an upper bound for
\begin{equation}\label{Uterm}
\int\displaylimits_{\mathbb{R}^2}Ue^{-\frac{k_{11}}{2k_{12}}\zeta'}dx.
\end{equation}
To this end, we use the inequality
\begin{equation}
h_0^{-1}U=\abs{x}^{-\kappa}e^{-\abs{x}^2}e^{u_0-\frac{k_{11}}{2k_{12}}v_3}\leq C_2,
\end{equation}
and we are now able to obtain the following estimate for (\ref{Uterm}),
\begin{equation}\label{lam2}
\nonumber\int\displaylimits_{\mathbb{R}^2}Ue^{-\frac{k_{11}}{2k_{12}}\zeta'}dx=\int\displaylimits_{\mathbb{R}^2}h_0^{-1}Ue^{-\frac{k_{11}}{2k_{12}}\zeta'}d\mu
\leq C_3\int\displaylimits_{\mathbb{R}^2}e^{-\frac{k_{11}}{2k_{12}}\zeta'}d\mu
\leq C_4(\varepsilon)\exp\left[\frac{k_{11}^2}{16k_{12}^2(4\pi-\varepsilon)}\|\nabla \zeta'\|_{L^2(dx)}^2\right].
\end{equation}
From (\ref{lam1}) and  (\ref{lam2}), we are able to obtain the following estimate for (\ref{lam}),
\begin{equation}\label{lamlow}
\Lambda=-2\pi\alpha\bar{\xi}+2\pi\beta\frac{\abs{K}}{k_{12}^2}\bar{\zeta}\geq C_5+\frac{\pi k_{11}}{4k_{12}(4\pi-\varepsilon)}\left(\beta\frac{\abs{K}}{k_{12}^2}-\alpha\right)\|\nabla \zeta'\|_{L^2(dx)}^2.
\end{equation}
We now find an estimate for the middle term in (\ref{func}).  We recall again that $f$ and $h$ are functions with compact support.  In the following, we use Young's inequality and the Poincar\'e inequality to obtain,
\begin{equation}\label{fmax}
\int\displaylimits_{\mathbb{R}^2}\abs{f\xi'}dx=\int\displaylimits_{\mathbb{R}^2}\abs{\frac{f}{\sqrt{2\varepsilon h_0}}\sqrt{2\varepsilon h_0}\xi'}dx
\leq\int\displaylimits_{\mathbb{R}^2}\frac{1}{4\varepsilon h_0}\abs{f}^2+\varepsilon h_0\abs{\xi'}^2dx
\leq \varepsilon^{-1} C_6 +\varepsilon C\| \nabla \xi'\|_{L^2(dx)}^2.
\end{equation}
Similarly, for the $\zeta '$ term we obtain,
\begin{equation}\label{hmax}
\int\displaylimits_{\mathbb{R}^2}\abs{h\zeta'}dx\leq\varepsilon^{-1}C_7+\varepsilon C\|\nabla\zeta'\|_{L^2(dx)}^2. 
\end{equation}
Therefore, in view of (\ref{lamlow}), (\ref{fmax}) and (\ref{hmax}), we have the following lower bound for $I(\xi,\zeta)$,
\begin{align}\label{lower}
\nonumber I(\xi,\zeta)&\geq \left(\frac{1}{2}-\varepsilon C\right)\|\nabla\xi'\|_{L^2(dx)}^2+\frac{\abs{K}}{2k_{12}^2}\left(1+\frac{\pi k_{11}}{2k_{12}(4\pi-\varepsilon)}\left(\beta-\frac{\alpha k_{12}^2}{\abs{K}} \right)-\varepsilon C'\right)\|\nabla\zeta'\|_{L^2(dx)}^2+C_8\\
&\equiv \delta_1\|\nabla\xi'\|_{L^2(dx)}^2+\delta_2\|\nabla\zeta'\|_{L^2(dx)}+C_8.
\end{align}
Since we know that $k_{12}<0$, and the constraints (\ref{constraints}) and (\ref{constraints1}) require $0<\beta<\alpha k_{12}^2/\abs{K}$, we see that 
\begin{equation}\label{beta1}
1+\frac{\pi k_{11}}{2k_{12}(4\pi-\varepsilon)}\left(\beta-\frac{\alpha k_{12}^2}{\abs{K}} \right)>0.
\end{equation}
Therefore, we can choose any $0<\varepsilon<4\pi$  so that $\delta_1,\delta_2>0$.  From this choice of $\varepsilon$, we see that
\begin{equation}\label{coercive}
I(\xi,\zeta)\rightarrow\infty\qquad\text{ when }\qquad\|\nabla \xi'\|_{L^2(dx)},\|\nabla \zeta'\|_{L^2(dx)}\rightarrow\infty
\end{equation}
which establishes that $I(\xi,\zeta)$ is coercive. We have thus established that $I(\xi,\zeta)$ is bounded below by (\ref{lower}) and coercive by (\ref{coercive}) on the admissible set
\begin{equation}
\mathcal{A}=\left\{\xi,\zeta\in\mathscr{H}| \xi,\zeta\text{ satisfy }(\ref{constraints})\right\}.
\end{equation}

We now pick a minimizing sequence of (\ref{min}), denoted by $\left\{(\xi_j,\zeta_j)\right\}_{j=1}^{\infty}$.  We note that since this is a minimizing sequence, and we have established coercivity, $I(\xi_j,\zeta_j)$ is bounded.  Therefore, both $\|\nabla \xi_j'\|_{L^2(dx)}^2$ and $\| \nabla \zeta_j'\|_{L^2(dx)}^2$ are bounded.  Consequently, the sequence $\left\{(\xi'_j,\zeta'_j)\right\}_{j=1}^{\infty}$ is bounded in $\tilde{\mathscr{H}}$.  With this, and the bounds for $\bar{\xi}$ and $\bar{\zeta}$, we can see that $\left\{\bar{\xi_j}\right\}_{j=1}^{\infty}$ and $\left\{\bar{\zeta_j}\right\}_{j=1}^{\infty}$ are bounded sequences in $\mathbb{R}$.  Without loss of generality, passing to a subsequence if necessary, we assume that there exist $\xi,\zeta\in\mathscr{H}$ so that
\begin{equation}
\xi_j\rightharpoonup\xi\qquad\text{and}\qquad\zeta_j\rightharpoonup\zeta, \text{ weakly in }\mathscr{H}.
\end{equation}
Now we show that the pair $(\xi,\zeta)$ satisfies the constraints (\ref{constraints}) and (\ref{constraints1}).  To this end, we use the mean value theorem, and H\"older's inequality to obtain
\begin{align}
\nonumber&\abs{\int\displaylimits_{\,\mathbb{R}^2}Ue^{-\frac{k_{11}}{2k_{12}}\zeta}dx-\int\displaylimits_{\mathbb{R}^2}Ue^{-\frac{k_{11}}{2k_{12}}\zeta_j}dx}\\
\nonumber&\leq C' \int\limits _{\mathbb{R}^2}h_0^{\frac{1}{2}}e^{-\frac{k_{11}}{2k_{12}}\abs{\zeta}}h_0^{\frac{1}{2}}e^{-\frac{k_{11}}{2k_{12}}\abs{\zeta_j}}\abs{\zeta-\zeta_j}dx\\ 
\nonumber&\leq C'\left(\int\displaylimits_{\,\,\mathbb{R}^2}e^{-\frac{k_{11}}{k_{12}}\abs{\zeta}}d\mu\right)^{\frac{1}{2}}\left(\int\displaylimits_{\,\,\mathbb{R}^2}e^{-\frac{k_{11}}{k_{12}}\abs{\zeta_j}}\abs{\zeta-\zeta_j}^2d\mu\right)^{\frac{1}{2}}\\
&\label{aa1}\leq C''(\varepsilon)e^{\frac{k_{11}^2}{8k_{12}^2(4\pi-\varepsilon)}\left(\|\nabla\zeta\|_{L^2(dx)}^2+\|\nabla\zeta_j \|_{L^2(dx)}^2 \right)}\|\zeta-\zeta_j\|_{L^2(d\mu)}.
\end{align}
The inequality (\ref{aa1}) comes from Lemma \ref{tm} with $a=-{k_{11}}/{k_{12}}$.  By Lemma \ref{embed}, we see that 
\begin{equation}
\|\zeta-\zeta_j\|_{L^2(d\mu)}\rightarrow 0\qquad \text{as}\qquad j\rightarrow\infty,
\end{equation}
and therefore,
\begin{equation}
\abs{\int\displaylimits_{\,\,\mathbb{R}^2}Ue^{-\frac{k_{11}}{2k_{12}}\zeta}dx-\int\displaylimits_{\mathbb{R}^2}Ue^{-\frac{k_{11}}{2k_{12}}\zeta_j}dx}\rightarrow 0\qquad\text{ as }j\rightarrow \infty.
\end{equation}
Similarly, we use the mean value theorem and H\"older's inequality to obtain,
\begin{equation}
\abs{\int\displaylimits_{\,\,\mathbb{R}^2}Ve^{\frac{1}{2}(\xi-\zeta)}dx-\int\displaylimits_{\mathbb{R}^2}Ve^{\frac{1}{2}(\xi_j-\zeta_j)}dx}\rightarrow 0\qquad\text{as}\qquad j\rightarrow\infty.
\end{equation}
Therefore, the pair $(\xi,\zeta)$ satisfies the constraints (\ref{constraints}). 
Furthermore, the last two terms of the functional, given by (\ref{func}), are continuous. To see this, we use $(\cdot,\cdot)$ to denote the $L^2$ inner product over $\mathbb{R}^2$ and also use the Cauchy-Schwarz inequality to get, 
\begin{equation}
\abs{(f,\xi)-(f,\xi_j)}=\abs{(f,\xi-\xi_j)}\leq \|f\|_{L^2(dx)}\|\xi-\xi_j\|_{L^2(dx)}.
\end{equation}
The same follows for the $h\zeta$ term. Since $f$ and $h$ are of compact support and continuous, they are bounded in $L^2$.  We see that the right hand side vanishes as $j\rightarrow\infty$.
Therefore, $I(\xi,\zeta)$ is weakly lower semicontinuous.  In other words,
\begin{equation}
I(\xi,\zeta)\leq \liminf_{j\rightarrow\infty}I(\xi_j,\zeta_j),
\end{equation}
and the pair $(\xi,\zeta)$ solves (\ref{main}).
We have thus established existence of a solution to the constrained minimization problem (\ref{min}) when $k_{12}<0$ under the condition given in (\ref{beta1}), $0<\beta<{\alpha k_{12}^2}/{\abs{K}}$ or in terms of $p$ and $q$,
\begin{equation}
0<\beta<\frac{\alpha}{4}\left(\frac{p}{q}+\frac{q}{p}-2\right).
\end{equation}
The first part of Theorem \ref{mainthm} has been established.
\end{proof}

\section{Constrained minimization with $0<k_{12}<1$}\label{min0<k12<1}

Here, we consider the case when $K$ is positive definite with $0<k_{12}<1$, and show that the minimization problem (\ref{min}) has a solution.  The approach of this section will mirror most of the previous section except for a small portion.  We will omit the details of repetition.

We state the theorem of this section below.

\begin{thm}\label{51thm}
The constrained minimization problem (\ref{min}) has a solution when $1<k_{12}<0$ and
\begin{equation*}
\beta>\frac{\alpha}{4}\left(\frac{p}{q}+\frac{q}{p}-2\right),\qquad\alpha>0.
\end{equation*}
\end{thm}
\begin{proof}
The approach mirrors that in the previous section until we get to the estimates for $\Lambda$ given by (\ref{lam}).  The only part that needs to be modified is how we establish a lower bound for
\begin{equation}
\bar{\zeta}=\frac{2 k_{12}}{k_{11}} \ln\left( \frac{\pi\beta}{k_{12}} -\frac{\pi\alpha k_{12}}{\abs{K}}\right)+\frac{2k_{12}}{k_{11}}\ln\left[\int\displaylimits_{\,\,\mathbb{R}^2}Ue^{-\frac{k_{11}}{2k_{12}}\zeta'} dx\right].
\end{equation}
Since ${2k_{12}}/{k_{11}}>0$, rather than finding an upper estimate, we need to find a lower estimate for
\begin{equation}
\int\displaylimits_{\mathbb{R}^2}Ue^{-\frac{k_{11}}{2k_{12}}\zeta'}dx.
\end{equation}
We recall that 
\begin{equation}
U=e^{u_0-\abs{x}^2-\frac{k_{11}}{2k_{12}}v_3},
\end{equation}
and $u_0, v_3,$ and $\abs{x}^2$ all belong to $L(d\mu)$.  We use the fact that $\zeta'\in\tilde{\mathscr{H}}$, along with Jensen's inequality to get
\begin{align}
\nonumber\int\displaylimits_{\mathbb{R}^2}Ue^{-\frac{k_{11}}{2k_{12}}\zeta'}dx &\geq\varepsilon_0\int\displaylimits_{\mathbb{R}^2}\left(Ue^{-\frac{k_{11}}{2k_{12}}\zeta'}\right)d\mu\\
\nonumber&=\varepsilon_0 \int\displaylimits_{\mathbb{R}^2}\left(e^{u_0-\frac{k_{11}}{2k_{12}}v_3-\abs{x}^2-\frac{k_{11}}{2k_{12}}\zeta'}\right)d\mu\\
&\geq C_1\exp\left[\int\displaylimits_{\,\,\mathbb{R}^2}\left(u_0-\frac{k_{11}}{2k_{12}}v_3-\abs{x}^2\right)d\mu\middle/\int\displaylimits_{\mathbb{R}^2}d\mu  \right].
\end{align}
We have now obtained a new lower bound for $\Lambda$,
\begin{equation}
\Lambda\geq C_2.
\end{equation}
Therefore, we are able to write the lower bound of $I(\xi,\zeta)$ as
\begin{equation}\label{ilower}
I(\xi,\zeta)\geq\left(\frac{1}{2}-\varepsilon C\right)\left(\| \nabla\xi'\|_{L^2(dx)}^2+\frac{\abs{K}}{k_{12}^2}\| \nabla\zeta'\|_{L^2(dx)}^2 \right)+C_3.
\end{equation}

We note here that the coefficients of $\|\nabla\xi'\|$ and $\|\nabla\zeta'\|$ in (\ref{ilower}) are positive.  Furthermore, the inequality is independent of $\alpha$ and $\beta$.  We now choose $\alpha$ and $\beta$ such that the constraints of the system when $0<k_{12}<1$ given in (\ref{constraints}) are satisfied, then we may choose $\varepsilon>0$ small enough so that $\frac{1}{2}-\varepsilon C>0$ and $I(\xi,\zeta)$ is coercive.  In the process, we have also shown that $I(\xi,\zeta)$ is bounded below.  

Therefore, we may now prove that the minimization problem has a solution. However, the rest of this proof is carried out in the same manner as the rest of Theorem \ref{k12min}.  The end result is that we established existence of a solution to the constrained minimization problem (\ref{min}) when $k_{12}>0$ under the condition, $\beta>{\alpha k_{12}^2}/{\abs{K}}$ or in terms of $p$ and $q$,
\begin{equation}
\beta>\frac{\alpha}{4}\left(\frac{p}{q}+\frac{q}{p}-2\right).
\end{equation}
This completes the proof of Theorem \ref{mainthm} 
\end{proof}

\section{Mountain pass for $k_{11},k_{12}>0$}\label{mpass}
In this section, we provide an alternative approach for establishing existence of solutions.  Here, we establish the existence of a saddle point type solution to the system via a mountain pass theorem.  As we will see, the result will provide a smaller range for $\alpha$ and $\beta$.  Before proceeding, we consider a modified version of the functional, $I(\xi,\zeta)$ used in the previous sections,
\begin{multline}\label{mpfun}
E(\xi,\zeta)=\int\displaylimits_{\mathbb{R}^2}\frac{1}{2}\abs{\nabla\xi}^2+\frac{\abs{K}}{2k_{12}^2}\abs{\nabla\zeta}^2+\frac{4\abs{K}}{k_{11}}Ve^{\frac{1}{2}(\xi-\zeta)}\\+\frac{4\abs{K}}{k_{11}}Ue^{-\frac{k_{11}}{2k_{12}}\zeta}+f\xi+\frac{\abs{K}}{k_{12}^2}h\zeta-\frac{4\abs{K}}{k_{11}}(V+U) dx.
\end{multline}
 We see that $E(0,0)=0$ and $E(c,c)\rightarrow -\infty$ as $c\rightarrow\infty$.
 We first establish that the functional (\ref{mpfun}) satisfies a certain compactness condition.  This lemma is stated below.
\begin{lem}
The functional given by (\ref{mpfun}), satisfies the Palais-Smale (PS) compactness condition.  That is, for every sequence $\left\{(\xi_n,\zeta_n)\right\}_{n=1}^{\infty}\subset\mathscr{H}\times\mathscr{H}$ such that $E(\xi_n,\zeta_n)\rightarrow M$ and $E'(\xi_n,\zeta_n)\rightarrow 0$ as $n\rightarrow\infty$, there exists a strongly convergent subsequence.
\end{lem}

\begin{proof}
First, we represent the PS conditions as,
\begin{multline}\label{ps1}
E(\xi_n,\zeta_n)=\int\displaylimits_{\mathbb{R}^2}\frac{1}{2}\abs{\nabla\xi_n}^2+\frac{\abs{K}}{2k_{12}^2}\abs{\nabla\zeta_n}^2+\frac{4\abs{K}}{k_{11}}Ve^{\frac{1}{2}(\xi-\zeta_n)}+\frac{4\abs{K}}{k_{11}}Ue^{-\frac{k_{11}}{2k_{12}}\zeta_n}\\
+f\xi+\frac{\abs{K}}{k_{12}^2}h\zeta_n-\frac{4\abs{K}}{k_{11}}(V+U)dx\rightarrow M,\qquad n\rightarrow\infty
\end{multline}
and  for any $w_1,w_2\in\mathscr{H}$,
\begin{multline}\label{ps2}
\abs{E'(\xi_n,\zeta_n)(w_1,w_2)}=
\Bigg|\int\displaylimits_{\mathbb{R}^2}\Bigg(\nabla\xi_n\cdot\nabla w_1+\frac{\abs{K}}{k_{12}^2}\nabla\zeta_n\cdot\nabla w_2+\frac{2\abs{K}}{k_{11}}Ve^{\frac{1}{2}(\xi_n-\zeta_n)}(w_1-w_2)\\-\frac{2\abs{K}}{k_{12}} Ue^{-\frac{k_{11}}{2k_{12}}\zeta_n}w_2+f w_1+\frac{\abs{K}}{k_{12}^2}hw_2\Bigg)dx\Bigg|
\leq\varepsilon_n\left(\|w_1\|_{\mathscr{H}}+\|w_2\|_{\mathscr{H}} \right).
\end{multline}
Where $E'$ represents the Fr\'echet derivative of (\ref{mpfun}).  We then choose $(w_1,w_2)=(1,0)$ and $(0,1)$ respectively in (\ref{ps2}) to obtain the following estimates,
\begin{equation}\label{w1in}
\abs{\int\displaylimits_{\,\,\mathbb{R}^2}\Bigg(\frac{2\abs{K}}{k_{11}}Ve^{\frac{1}{2}(\xi_n-\zeta_n)}+f\Bigg)dx}\leq\varepsilon_n\left(\int\displaylimits_{\,\,\mathbb{R}^2} d\mu\right)^{\frac{1}{2}},
\end{equation}
and
\begin{equation}\label{w2in}
\abs{\int\displaylimits_{\,\,\mathbb{R}^2}\Bigg(\frac{-2\abs{K}}{k_{11}}Ve^{\frac{1}{2}(\xi_n-\zeta_n)}-\frac{2\abs{K}}{k_{12}}Ue^{-\frac{k_{11}}{2k_{12}}\zeta_n}+\frac{\abs{K}}{k_{12}^2}h\Bigg)dx}\leq\varepsilon_n\left(\int\displaylimits_{\,\,\mathbb{R}^2} d\mu\right)^{\frac{1}{2}}.
\end{equation}
Again, we use the decomposition of $\xi=\bar{\xi}+\xi '$ and $\zeta=\bar{\zeta}+\zeta '$, insert this into inequalities (\ref{w1in}) and (\ref{w2in}) to get,
\begin{equation}\label{a1}
\abs{\int\displaylimits_{\,\,\mathbb{R}^2}\Bigg(\frac{2\abs{K}}{k_{11}}e^{\frac{1}{2}(\bar{\xi}_n-\bar{\zeta}_n)}Ve^{\frac{1}{2}(\xi_n'-\zeta_n')}+f\Bigg)dx}\leq\varepsilon_n\left(\int\displaylimits_{\,\,\mathbb{R}^2} d\mu\right)^{\frac{1}{2}},
\end{equation}
and
\begin{equation}\label{a2}
\abs{\int\displaylimits_{\,\,\mathbb{R}^2}\Bigg( -\frac{2\abs{K}}{k_{11}}Ve^{\frac{1}{2}(\bar{\xi}_n-\bar{\zeta}_n)}Ve^{\frac{1}{2}(\xi_n'-\zeta_n')}-\frac{2\abs{K}}{k_{12}}e^{-\frac{k_{11}}{2k_{12}}\bar{\zeta}_n} U e^{-\frac{k_{11}}{2k_{12}}\zeta_n '}\Bigg)  dx}\leq\varepsilon_n\left(\int\displaylimits_{\,\,\mathbb{R}^2} d\mu\right)^{\frac{1}{2}}.
\end{equation}
The next step is to  show that $\bar{\xi}_n$ and $\bar{\zeta}_n$ are bounded in $\mathbb{R}$.  We also need to show that $\xi_n ',\zeta_n '$ are bounded in $\tilde{\mathscr{H}}$.   To this end, in (\ref{a1}), we solve for $e^{\frac{1}{2}(\bar{\xi}_n-\bar{\zeta}_n)}$ and obtain,
\begin{align}\label{a3}
\nonumber&\frac{k_{11}}{2\abs{K}}\left(2\pi\alpha-\varepsilon_n\left(\int\displaylimits_{\,\,\mathbb{R}^2} d\mu\right)^{\frac{1}{2}}\right)\left(\int\displaylimits_{\,\,\mathbb{R}^2} Ve^{\frac{1}{2}(\xi_n '-\zeta_n ')}dx \right)^{-1}\\\nonumber&\leq e^{\frac{1}{2}(\bar{\xi}_n-\bar{\zeta}_n)}\\&\leq\frac{k_{11}}{2\abs{K}}\left(2\pi\alpha+\varepsilon_n\left(\int\displaylimits_{\,\,\mathbb{R}^2} d\mu\right)^{\frac{1}{2}}\right)\left(\int\displaylimits_{\,\,\mathbb{R}^2} Ve^{\frac{1}{2}(\xi_n '-\zeta_n ')}dx \right)^{-1}.
\end{align}
We then subtract (\ref{a1}) from (\ref{a2}) and isolate the term $e^{-\frac{k_{11}}{2k_{12}}\bar{\zeta}_n}$ to obtain,
\begin{align}\label{a4}
\nonumber&-\frac{k_{12}}{2\abs{K}}\left(2\pi\alpha-\frac{2\pi\beta\abs{K}}{k_{12}^2} +2\varepsilon_n\left(\int\displaylimits_{\,\,\mathbb{R}^2} d\mu\right)^{\frac{1}{2}}\right)\\\nonumber&\leq e^{-\frac{k_{11}}{2k_{12}}\bar{\zeta}_n}\\& \leq-\frac{k_{12}}{2\abs{K}}\left(2\pi\alpha-\frac{2\pi\beta\abs{K}}{k_{12}^2} -2\varepsilon_n\left(\int\displaylimits_{\,\,\mathbb{R}^2} d\mu\right)^{\frac{1}{2}}\right).
\end{align}
We see that the above inequalities (\ref{a3}) and (\ref{a4}) are valid under the following conditions for $\alpha$ and $\beta$,
\begin{equation}\label{beta}
\beta>\frac{k_{12}^2\alpha}{\abs{K}}.
\end{equation}
We then take the natural logarithm of (\ref{a4}) and isolate the $\bar{\zeta}_n$ term.  We use Jensen's inequality and the Trudinger-Moser inequality to obtain the following,
\begin{equation}
C_1\leq\bar{\zeta}_n\leq C_2(\varepsilon)+\frac{k_{11}}{8k_{12}(4\pi-\varepsilon)}\|\nabla\zeta'\|_{L^2(dx)}^2.
\end{equation}
In order to find estimates for $\bar{\xi}_n$, we will use H\"older's inequality to obtain the following for $s_1,s_2>1$ with ${1}/{s_1}+{1}/{s_2}=1$ to obtain, 
\begin{equation}
\int\displaylimits_{\mathbb{R}^2}Ve^{\frac{1}{2}(\xi_n'-\zeta_n')}dx\leq\int\displaylimits_{\mathbb{R}^2}Ve^{\frac{1}{2}\abs{\xi_n'}+\frac{1}{2}\abs{\zeta_n'}}dx
\leq C\left(\int\displaylimits_{\,\,\mathbb{R}^2}e^{\frac{s_1}{2}\abs{\xi_n'}}d\mu\right)^{\frac{1}{s_1}}\left( \int\displaylimits_{\,\,\mathbb{R}^2}e^{\frac{s_2}{2}\abs{\zeta_n'}}d\mu\right)^{\frac{1}{s_2}}.
\end{equation}
Furthermore, by using the Trudinger-Moser inequality, we see that
\begin{equation}
C\left(\int\displaylimits_{\,\,\mathbb{R}^2}e^{\frac{s_1}{2}\abs{\xi_n'}}d\mu\right)^{\frac{1}{s_1}}\left( \int\displaylimits_{\,\,\mathbb{R}^2}e^{\frac{s_2}{2}\abs{\zeta_n'}}d\mu\right)^{\frac{1}{s_2}}
\leq C'(\varepsilon)e^{\frac{1}{16(4\pi-\varepsilon)}\left(s_1\|\nabla\xi_n'\|_{L^2(dx)}^2+s_2\|\nabla\zeta_n'\|_{L^2(dx)}^2 \right)}.
\end{equation}
Therefore, we have the following estimate for $\bar{\xi}_n$,
\begin{equation}
\bar{\zeta_n}+C_3-\frac{1}{16(4\pi-\varepsilon)}\left(s_1\|\nabla\xi_n'\|_{L^2(dx)}^2+s_2\|\nabla\zeta_n'\|_{L^2(dx)}^2\right)\leq\bar{\xi}_n\leq C_4+\bar{\zeta}_n.
\end{equation}
We see that the boundedness of $\bar{\xi}_n$ and $\bar{\zeta}_n$ depend upon the gradient terms, $\|\nabla\xi'\|$ and $\|\nabla\zeta'\|$.  In order to establish the bounds of these gradient terms, we take $n$ to be large enough in (\ref{ps1}) and use the decompositions $\xi_n=\xi_n'+\bar{\xi}_n$, $\zeta_n=\zeta_n'+\bar{\zeta}_n$ so that
\begin{multline}\label{b1}
\abs{M}+1\geq\int\displaylimits_{\mathbb{R}^2}\Bigg(\frac{1}{2}\abs{\nabla\xi_n'}^2+\frac{\abs{K}}{2k_{12}^2}\abs{\nabla\zeta_n'}^2dx+\frac{4\abs{K}}{k_{11}}e^{\frac{1}{2}(\bar{\xi}_n-\bar{\zeta}_n)}\int\displaylimits_{\mathbb{R}^2}Ve^{\frac{1}{2}(\xi_n'-\zeta_n')}\Bigg)dx+\\\frac{4\abs{K}}{k_{11}}e^{-\frac{k_{11}}{2k_{12}}\bar{\zeta}_n}\int\displaylimits_{\mathbb{R}^2}Ue^{-\frac{k_{11}}{2k_{12}}\zeta_n'}dx+\int\displaylimits_{\mathbb{R}^2}f\xi_n'dx+\frac{\abs{K}}{2k_{12}^2}\int\displaylimits_{\mathbb{R}^2}h\zeta_n'dx\\-2\pi\alpha\bar{\xi}_n+2\pi\beta\frac{\abs{K}}{k_{12}^2}\bar{\zeta}_n-\frac{4\abs{K}}{k_{11}}\left(U'+V'\right),
\end{multline}
where we have used the fact that  $U$ and $V$ are bounded above by positive constants $U'$ and $V'$ respectively.  We recall the estimates for the $f$ and $h$ terms,
\begin{equation*}
\int\displaylimits_{\mathbb{R}^2}=\abs{f\xi'}dx\leq \varepsilon^{-1}C_6+\varepsilon C\|\nabla\xi'\|_{L^2(dx)}^2\qquad\text{and}\qquad
\int\displaylimits_{\mathbb{R}^2}=\abs{h\zeta'}dx\leq \varepsilon^{-1}C_7+\varepsilon C\|\nabla\zeta'\|_{L^2(dx)}^2,
\end{equation*}
and use our estimates for $\bar{\xi}_n$ and $\bar{\zeta}_n$ in (\ref{b1}) along with the Trudinger-Moser inequality to obtain,
\begin{multline}
\abs{M}+1\geq \left(\frac{1}{2}-\varepsilon C\right)\|\nabla\xi_n'\|_{L^2(dx)}^2+\frac{\abs{K}}{k_{12}^2}\left(\frac{1}{2}-\varepsilon C-\frac{k_{11}k_{12}\pi\alpha}{4\abs{K}(4\pi-\varepsilon)} \right)\|\nabla\zeta_n'\|_{L^2(dx)}^2\\-2\pi\alpha C_5(\varepsilon)-\varepsilon^{-1}C_6-\frac{4\abs{K}}{k_{11}}\left(U'+V'\right).
\end{multline}
We let
\begin{align}
\delta_1&=\frac{1}{2}-\varepsilon C,\qquad\delta_2 =\frac{\abs{K}}{K_{12}^2}\left(\frac{1}{2}-\varepsilon C-\frac{k_{11}k_{12}\pi\alpha}{4\abs{K}(4\pi-\varepsilon)}\right)
\end{align}
and obtain,
\begin{equation}
\abs{M}+C_7+\frac{4\abs{K}}{k_{11}}\left(U'+V'\right)\geq\delta_1\|\nabla\xi_n'\|_{L^2(dx)}^2+\delta_2\|\nabla\zeta_n'\|_{L^2(dx)}^2
\end{equation}
In order to ensure that we may choose $\varepsilon>0$ small enough so that $\delta_1,\delta_2>0$, we require
\begin{equation}
\alpha<\frac{8\abs{K}}{k_{11}k_{12}}
\end{equation}
and recall the conditions on $\beta$,
\begin{equation*}
\beta>\frac{k_{12}^2\alpha}{\abs{K}}.
\end{equation*}
This can be written in terms of $p$ and $q$,
\begin{equation}\label{pq1}
\alpha <32\left(\frac{p}{q}-\frac{q}{p}\right)^{-1}\qquad\text{and}\qquad\beta>\frac{\alpha}{4}\left(\frac{p}{q}+\frac{q}{p}-2\right).
\end{equation}
We recall that $k_{11},k_{12}>0$, and $\abs{K}>0$.  Therefore, $p>q$ and the condition (\ref{pq1}) is valid.

We have now established that $\left\{\nabla\xi_n'\right\}$, $\left\{\nabla\zeta_n'\right\}$ are bounded in $L^2(dx)$ and consequently that $\left\{\bar{\xi}_n\right\},\left\{\bar{\zeta}_n\right\}$ are bounded in $\mathbb{R}$.  From the Poincar\'e inequality, we see that both $\xi_n'$ and $\zeta_n'$ are bounded above and thus the sequence $\left\{(\xi_n,\zeta_n)\right\}$ is bounded in $\mathscr{H}\times\mathscr{H}$. Therefore, there exists a subsequence converging weakly in  $\mathscr{H}\times\mathscr{H}$.  Passing to a subsequence if necessary, we write this as,
\begin{equation*}
\xi_n\rightharpoonup \xi\qquad\text{and}\qquad\zeta_n\rightharpoonup\zeta.
\end{equation*}
Since $\mathscr{H} \subset\subset L^2(d\mu)$, we have strong convergence in $L^2(d\mu)$ or in other words,
\begin{equation}
\|\xi_n-\xi\|_{L^2(d\mu)},\qquad\|\zeta_n-\zeta\|_{L^2(d\mu)}\rightarrow0,\qquad \text{ as }n\rightarrow\infty.
\end{equation}
Now we need to establish strong convergence in $\mathscr{H}$.  To this end, we take the limit as $n\rightarrow\infty$ in (\ref{ps2}) and obtain,
\begin{equation}\label{pslim}
\int\displaylimits_{\mathbb{R}^2}\Bigg(\nabla\xi\cdot\nabla w_1+\frac{\abs{K}}{k_{12}^2}\nabla\zeta\cdot\nabla w_2 +\frac{2\abs{K}}{k_{11}}Ve^{\frac{1}{2}(\xi-\zeta)}(w_1-w_2)-\frac{2\abs{K}}{k_{12}}Ue^{-\frac{k_{11}}{2k_{12}}\zeta}w_2+fw_1+\frac{\abs{K}}{k_{12}^2}hw_2d\Bigg)x=0.
\end{equation}
We now subtract (\ref{pslim}) from (\ref{ps2}) and set $w_1=\xi_n-\xi,w_2=\zeta_n-\zeta$ to get,
\begin{multline}\label{mpsubtract}
\Bigg| \int\displaylimits_{\mathbb{R}^2}\Bigg(\abs{\nabla(\xi_n-\xi)}^2+\frac{\abs{K}}{k_{12}^2}\abs{\nabla(\zeta_n-\zeta)}^2+\frac{2\abs{K}}{k_{11}}V\left(e^{\frac{1}{2}(\xi_n-\zeta_n)}-e^{\frac{1}{2}(\xi-\zeta)} \right)\left(\xi_n-\xi+\zeta-\zeta_n \right)\\
+\frac{2\abs{K}}{k_{12}^2}U\left(e^{-\frac{k_{11}}{2k_{12}}\zeta}-e^{-\frac{k_{11}}{2k_{12}}\zeta_n} \right)\left(\zeta_n-\zeta \right)\Bigg)dx\Bigg| \leq\varepsilon_n\left(\|\xi_n-\xi\|_{\mathscr{H}}+\|\zeta_n-\zeta \|_{\mathscr{H}} \right).
\end{multline}
We can rewrite (\ref{mpsubtract}) and use the norm notation to obtain,
\begin{multline}
\|\nabla(\xi_n-\xi)\|_{L^2(dx)}^2+\frac{\abs{K}}{k_{12}^2}\|\nabla(\zeta_n-\zeta)\|_{L^2(dx)}^2\leq \frac{2\abs{K}}{k_{11}}\int\displaylimits_{\mathbb{R}^2}\abs{Ve^{\frac{1}{2}(\xi-\zeta)}-Ve^{\frac{1}{2}(\xi_n-\zeta_n)}}\abs{\xi_n-\xi}dx\\
+\frac{2\abs{K}}{k_{11}}\int\displaylimits_{\mathbb{R}^2}\abs{Ve^{\frac{1}{2}(\xi-\zeta)}-Ve^{\frac{1}{2}(\xi_n-\zeta_n)}}\abs{\zeta_n-\zeta}dx+\frac{2\abs{K}}{k_{12}^2}\int\displaylimits_{\mathbb{R}^2}\abs{Ue^{-\frac{k_{11}}{2k_{12}}\zeta_n}-Ue^{-\frac{k_{11}}{2k_{12}}\zeta}}\abs{\zeta_n-\zeta}dx\\
+\varepsilon_n\left(\|\xi_n-\xi\|_{\mathscr{H}}+\|\zeta_n-\zeta\|_{\mathscr{H}}\right).
\end{multline}
We now write $dx=h_0^{-\frac{1}{2}}h_0^{\frac{1}{2}}dx$ for the integrals involving $V$ and use H\"older's inequality to obtain
\begin{multline}\label{psconv}
\|\nabla(\xi_n-\xi)\|_{L^2(dx)}^2+\frac{\abs{K}}{k_{12}^2}\|\nabla(\zeta_n-\zeta)\|_{L^2(dx)}^2\\\leq C_1\left(\int\displaylimits_{\,\,\mathbb{R}^2}\abs{e^{\frac{1}{2}(\xi-\zeta)}-e^{\frac{1}{2}(\xi_n-\zeta_n)}}^2dx \right)^{\frac{1}{2}}\left(\|\xi_n-\xi \|_{L^2(d\mu)}+\|\zeta_n-\zeta \|_{L^2(d\mu)} \right)\\
+C_2\left(\int\displaylimits_{\,\,\mathbb{R}^2}\abs{e^{-\frac{k_{11}}{2k_{12}}\zeta_n}-e^{-\frac{k_{11}}{2k_{12}}\zeta}}^2 \right)^{\frac{1}{2}}\|\zeta_n-\zeta\|_{L^2(d\mu)}
+\varepsilon_n\left(\|\xi_n-\xi\|_{\mathscr{H}}+\|\zeta_n-\zeta\|_{\mathscr{H}}\right).
\end{multline}
Since $\xi,\zeta,\xi_n,\zeta_n$ are bounded in $\mathscr{H}$, and the integrals of the exponential terms are bounded, we have that the right hand side vanishes as we let $n\rightarrow\infty$ in (\ref{psconv}).  This implies that
\begin{equation}
\xi_n\rightarrow\xi\qquad\text{and}\qquad\zeta_n\rightarrow\zeta \text{ strongly  in }\mathscr{H}
\end{equation}
and we have established our PS compactness condition.  Furthermore, we established that $(\xi,\zeta)$ is a critical point of $E(\xi,\zeta)$ by (\ref{pslim}).

\end{proof}
We now show the existence of a mountain pass structure.  However, we first note that both $U$ and $V$ are bounded below and denote their respective lower bounds by $U_0$ and $V_0$.
\begin{lem}
There exist constants $a,r>0$ such that $E(\xi,\zeta)\geq a$ for any $\xi,\zeta$ satisfying $\|\xi\|_{\mathscr{H}}+\|\zeta\|_{\mathscr{H}}=r$
\end{lem}
\begin{proof}
By the Poincar\'e inequality, we see that
\begin{equation}
C\|\nabla\xi\|_{L^2(dx)}^2+C\|\nabla\zeta\|_{L^2(dx)}^2\geq r^2.
\end{equation}
Then, we have the lower bound for $E(\xi,\zeta)$,
\begin{equation}
E(\xi,\zeta)\geq\left(\frac{1}{2}-\varepsilon C\right)\|\nabla\xi\|_{L^2(dx)}^2+\frac{\abs{K}}{k_{12}^2}\left(\frac{1}{2}-\varepsilon C\right)\|\nabla\zeta\|_{L^2(dx)}^2-\varepsilon^{-1}C_2 -\frac{4\abs{K}}{k_{11}}\left(V_0+U_0 \right)\left(\int\displaylimits_{\,\,\mathbb{R}^2}d\mu \right).
\end{equation}
We now let $\delta=\min\left\{\left(\frac{1}{2}-\varepsilon C\right),\frac{\abs{K}}{{k_{12}^2}}\left(\frac{1}{2}-\varepsilon C\right)\right\}$ and obtain,
\begin{equation}
E(\xi,\zeta)\geq \delta r^2-\varepsilon^{-1}C_2 -\frac{4\abs{K}}{k_{11}}\left(V_0+U_0 \right)\left(\int\displaylimits_{\,\,\mathbb{R}^2}d\mu \right).
\end{equation}
By choosing an $r$ big enough, we set $a= \delta r^2-\varepsilon^{-1}C_2 -{4\abs{K}}\left(V_0+U_0 \right)/{k_{11}}\left(\int\displaylimits_{\mathbb{R}^2}d\mu \right)>0$ and to this end we obtain $E(\xi,\zeta)\geq a$. Furthermore, we note that if we take $(\xi,\zeta)=(c,c)$, and take $c>r$ large enough, we obtain $E(c,c)<0$, since $E$ is indefinite.  
\end{proof}
We have now established all conditions for the mountain pass theorem, which we state below.

\begin{thm}\label{mountainthm}
For $0<\alpha<\frac{8\abs{K}}{k_{11}k_{12}}$ and $\beta>\frac{k_{12}^2}{\abs{K}}\alpha$, the functional (\ref{mpfun}) has a nontrivial critical point in $\mathscr{H}\times \mathscr{H}$.  This nontrivial critical point is a nonzero classical solution of the system (\ref{main}). 
\end{thm}

\begin{proof}
We assume that $\delta$ is such that the functional $E(\xi,\zeta)$ satisfies the Palais-Smale condition.  We now define 
\begin{equation}
\Gamma=\left\{g\in C\left(C\left[0,1 \right]; \mathscr{H}\times\mathscr{H}\right) | g(0)=(0,0), g(1)=(\xi,\zeta)  \right\}
\end{equation}
then there is a point $t_g\in(0,1)$ such that
\begin{equation}
\|g(t_g)\|_{\mathscr{H}\times\mathscr{H}}=R.
\end{equation}
Therefore, we have a critical value of the functional, (\ref{mpfun}),
\begin{equation}
c=\inf_{g\in\Gamma}\max_{0\leq t\leq 1} E(g(t))\geq a
\end{equation}
As a result, $E$ has a critical point, $(\xi,\zeta)\in\mathscr{H}\times\mathscr{H}$, such that $E(\xi,\zeta)=c$.  Since $c\geq a>0$, and $E(0,0)=0$, we see that $(\xi,\zeta)$ is nontrivial.

\end{proof}

In the next section, we discuss the rate of decay of solutions established in sections \ref{mink12<0}, \ref{min0<k12<1}, and \ref{mpass}.

\section{Asymptotic Decay Estimates}\label{decay}
In this section, we discuss the decay of the solutions obtained in the previous sections. First, by the substitutions in the work that precedes this section, we see that the solutions we established are of the form,
\begin{align}\label{sol1}
\begin{split}
u&=u_0-\abs{x}^2+\frac{k_{11}}{2k_{12}}(\zeta+v_3)\\
v&=v_0-\abs{x}^2+\frac{1}{2}(\xi+u_3)+\frac{1}{2}(\zeta+v_3).
\end{split}
\end{align}
We recall that in the derivation of the system, Medina \cite{medina} had used the substitution,
\begin{align}\label{sol2}
\begin{split}
u&=\ln\abs{\psi_{\uparrow}}^2-\ln\abs{\bar{\rho}},\\
v&=\ln\abs{\psi_{\downarrow}}^2-\ln\abs{\bar{\rho}}.
\end{split}
\end{align}
In view of (\ref{sol1}) and (\ref{sol2}), we see that the solutions to the original system are
\begin{align}
\begin{split}
\abs{\psi_{\uparrow}}^2&=\abs{\bar{\rho}}e^{u_0-\abs{x}^2}e^{\frac{k_{11}}{2k_{12}}(\zeta+v_3)},\\
\abs{\psi_{\downarrow}}^2&=\abs{\bar{\rho}}e^{v_0-\abs{x}^2}e^{\frac{1}{2}(\xi+u_3)}e^{\frac{1}{2}(\zeta+v_3)}.
\end{split}
\end{align}
Moreover, from the constraints given by (\ref{constraints}) and (\ref{constraints1}), we may obtain
\begin{align}
\int\displaylimits_{\mathbb{R}^2}\abs{\psi_{\uparrow}}^2dx&=\int\displaylimits_{\mathbb{R}^2}Ue^{-\frac{k_{11}}{2k_{12}}\zeta}dx=\frac{p\pi}{p-q}\left[\beta-\frac{\alpha}{4}\left(\frac{p}{q}+\frac{q}{p}-2\right)\right]\label{bps1}\\
\int\displaylimits_{\mathbb{R}^2}\abs{\psi_{\downarrow}}^2dx&=\int\displaylimits_{\mathbb{R}^2}Ve^{\frac{1}{2}(\xi-\zeta)}dx=\alpha\frac{16\pi q}{p+q}\label{bps2}.
\end{align}
We then look back at the energy of the system given by (\ref{energy}) and use (\ref{bps1}) along with (\ref{bps2}) to see that
\begin{align}\label{energy2}
E&=\frac{1}{2M}\int \displaylimits_{\,\,\mathbb{R}^2}\Bigg(\sum \abs{(D_1^{\sigma}-iD_2^{\sigma})\psi_{\sigma}}^2+eB(\abs{\psi_{\uparrow}}^2+\abs{\psi_{\downarrow}}^2)-eB(\abs{\psi_{\uparrow,0}}^2+\abs{\psi_{\downarrow,0}}^2)\Bigg)dx\\
&=\frac{eB}{2M}\left(\alpha\frac{16\pi q}{p+q}+\frac{p\pi}{p-q}\left[\beta-\frac{\alpha}{4}\left(\frac{p}{q}+\frac{q}{p}-2\right)\right]-\int\displaylimits_{\,\,\mathbb{R}^2}\bar{\rho}\right).
\end{align}
As we can see, the total energy diverges and different values of $\alpha$ and $\beta$ give rise to multivortex solutions of (\ref{bps}) that are of divergent energy.

We now study the decay rate of solutions as in \cite{mcowen1979behavior,mcowen1984equation,sy2}.  We define $W_{s,\delta}^2$  to be the closure of $C^{\infty}$ functions with compact support over $\mathbb{R}^2$ with the norm,
\begin{equation*}
\| u\|_{W_{s,\delta}^2}^2=\sum_{\abs{\gamma}\leq s}\|(1+\abs{x})^{\delta+\abs{\gamma}}D^{\gamma}u\|_{L^2(dx)}^2.
\end{equation*}
Furthermore, we let $C_0(\mathbb{R}^2)$ denote the set of continuous functions over $\mathbb{R}^2$ vanishing as $\abs{x}\rightarrow\infty$. Below we state three necessary lemmas, which were fully established in \cite{mcowen1979behavior,mcowen1984equation}.
\begin{lem}\label{wc0}
If $s>1$ and $\delta>-1$ then $W_{s,\delta}^2\subset C_0(\mathbb{R}^2).$
\end{lem}
\begin{lem}\label{oneto}
For $-1<\delta<0$, the Laplace operator $\Delta:W_{2,\delta}^2\rightarrow W_{0,\delta+2}^2$ is one-to-one.  Furthermore, the range of $\Delta$ has the following characterization,
\begin{equation*}
\Delta\left(W_{0,\delta+2}^2\right)=\left\{ f\in W_{0,\delta+2}^2 \Bigg| \int\displaylimits_{\mathbb{R}^2} fdx=0\right\}.
\end{equation*}
\end{lem}
\begin{lem}\label{constlem}
If $\xi\in\mathscr{H}$ and $\Delta\xi=0$ then $\xi=$constant.
\end{lem}

\begin{lem}
Let $(\xi,\zeta)$ be a solution pair of the system obtained as a minimizer of (\ref{min}).  Then $\xi$ and $\zeta$ both approach some constants at infinity.
\end{lem}
\begin{proof}
First, we recall the Trudinger-Moser type inequality \cite{mcowen1984equation},
\begin{equation*}
\int\displaylimits_{\mathbb{R}^2}e^{a\abs{u}}d\mu\leq C(\varepsilon)e^{\frac{a^2}{4(4\pi-\varepsilon)}\| \nabla u\|_{L^2(dx)}^2},\qquad u\in\tilde{\mathscr{H}}
\end{equation*}
and we look at the right hand side of (\ref{main}) and let 
\begin{equation}
f_1=\frac{2\abs{K}}{k_{11}}Ve^{\frac{1}{2}(\xi-\zeta)}+f\qquad\text{and}\qquad
h_1 =-2k_{12}Ue^{-\frac{k_{11}}{2k_{12}}\zeta}-2\frac{k_{12}^2}{k_{11}}Ve^{\frac{1}{2}(\xi-\zeta)}+h.
\end{equation}
Since we have the constraints (\ref{constraints}), we see that $f_1,h_1\in L^2(dx)$ with,
\begin{equation}
\int\displaylimits_{\mathbb{R}^2}f_1dx=\int\displaylimits_{\mathbb{R}^2}h_1dx=0.
\end{equation}

Now we recall that $\bar{\xi}$ and $\bar{\zeta}$ are bounded in $\mathbb{R}$.  Also, the functions $\nabla \xi'$ and $\nabla\zeta'$ are bounded in $L^2(dx)$.  With this in mind, we now claim that $f_1,h_1\in W_{0,\delta+2}^2$ for $-1<\delta<0$. To see this, we first observe that for this choice of $s$ and $\delta$, we have the following norm,
\begin{equation*}
\| u\|_{W_{0,\delta}^2}^2=\|(1+\abs{x})^{\delta+2}u\|_{L^2(dx)}^2.
\end{equation*}
Then, for $f_1$ we use the triangle inequality to get the following,
\begin{align}
\nonumber\|f_1\|_{W_{0,\delta+2}^2}^2&=\int\displaylimits_{\mathbb{R}^2}\abs{(1+\abs{x})^{\delta+2}f_1}^2dx\\
\nonumber&=\int\displaylimits_{\mathbb{R}^2}(1+\abs{x})^{2\delta+4}\abs{\frac{2\abs{K}}{k_{11}}Ve^{\frac{1}{2}(\xi-\zeta)}+f}^2dx\\
&\leq \int\displaylimits_{\mathbb{R}^2}(1+\abs{x})^{2\delta+4}\left(\frac{4\abs{K}^2}{k_{11}^2}V^2e^{\xi-\zeta}+\abs{\frac{4\abs{K}}{k_{11}}}\abs{f}Ve^{\frac{1}{2}(\xi-\zeta)}+\abs{f}^2 \right)dx.
\end{align}

Recall that $V=O(e^{-\abs{x}})$ and $f$ has compact support and satisfies $\int fdx=-2\pi\alpha$.  As a result, we see that the last term in on the right hand side is bounded.  We must then find bounds for the first two terms.  To this end, we use the decomposition $\xi=\bar{\xi}+\xi'$ and $\zeta=\bar{\zeta}+\zeta'$.  We observe that for the first term, $V^2$ controls  $(1+\abs{x})^{2\delta+4}$, and use the Trudinger-Moser type inequality to obtain,
\begin{align}
\int\displaylimits_{\mathbb{R}^2}\Bigg(\frac{4\abs{K}^2}{k_{11}^2}(1+\abs{x})^{2\delta+4}V^2e^{\xi-\zeta}\Bigg)dx&\leq C_1e^{\bar{\xi}-\bar{\zeta}}\int\displaylimits_{\mathbb{R}^2}e^{\xi'-\zeta'}d\mu\nonumber\\
&\leq C'(\varepsilon)e^{\frac{1}{4(4\pi-\varepsilon)}\left(\|\nabla \xi'\|_{L^2(dx)}^2+\|\nabla\zeta'\|_{L^2(dx)}^2\right)}\nonumber\\
&\leq C''.
\end{align}
Now we must find an upper bound for the second term.  We use Young's inequality with $\varepsilon$ to get the following,
\begin{align}
\nonumber\int\displaylimits_{\mathbb{R}^2}(1+\abs{x})^{2\delta+4}\abs{\frac{4\abs{K}}{k_{11}}}\abs{f}Ve^{\frac{1}{2}(\xi-\zeta)}dx&=\abs{\frac{4\abs{K}}{k_{11}}}\int\displaylimits_{\mathbb{R}^2}\Bigg((1+\abs{x})^{2\delta+4}\frac{\abs{f}}{\sqrt{2\varepsilon h_0}}Ve^{\frac{1}{2}(\xi-\zeta)}\sqrt{2\varepsilon h_0}\Bigg)dx\\
\nonumber&\leq C_2\int\displaylimits_{\mathbb{R}^2}\frac{(1+\abs{x})^{4\delta+8}}{4\varepsilon h_0}\abs{f}^2dx+C_2\varepsilon\int\displaylimits_{\mathbb{R}^2}Ve^{\xi-\zeta}d\mu\\
&\leq C_2'\varepsilon^{-1}+C''.
\end{align}
This establishes the fact that $f_1\in W_{0,\delta+2}^2$.  We now turn to $h_1$.  For this, we recall that $h$ has compact support with $\int hdx=2\pi\beta$ and $U=O(e^{-\abs{x}})$. We observe that,
\begin{align}
\nonumber\|h_1\|_{W_{0,\delta+2}^2}^2&=\|2k_{12}(1+\abs{x})^{\delta+2}Ue^{-\frac{k_{11}}{2k_{12}}\zeta}+2\frac{k_{12}^2}{k_{11}}(1+\abs{x})^{\delta+2}Ve^{\frac{1}{2}(\xi-\zeta)}-(1+\abs{x})^{\delta+2}h\|_{L^2(dx)}^2\\\nonumber
&\leq C_3\|(1+\abs{x})^{\delta+2}Ue^{-\frac{k_{11}}{2k_{12}}\zeta}\|_{L^2(dx)}^2\\&\quad+C_4\|(1+\abs{x})^{\delta+2}Ve^{\frac{1}{2}(\xi-\zeta)}\|_{L^2(dx)}^2+\|(1+\abs{x})^{\delta+2}h\|_{L^2(dx)}^2
\end{align}
The boundedness of the last two terms follows from the bound for $f_1$.  To establish a bound for the first term, we use the Trudinger-Moser inequality.  We have established our claim.  

We now use Lemma \ref{oneto} to obtain that there are unique $\xi_1,\zeta_1\in W_{2,\delta}^2$ so that 
\begin{equation*}
\Delta \xi_1=f_1\qquad\text{and}\qquad \Delta\zeta_1=h_1.
\end{equation*}
Now, by Lemma \ref{wc0}, we see that as $\abs{x}\rightarrow\infty$, then $\xi_1,\zeta_1\rightarrow 0$.
Moreover,  $\xi_1,\zeta_1\in L^2(d\mu)$. Also, since $\nabla\xi_1,\nabla\zeta_1\in W_{0,\delta+1}^2$ and $\delta>-1$  we see that
 \begin{equation}\label{w0}
 \|\nabla\xi_1\|_{W_{0,\delta+1}}^2=\|(1+\abs{x})^{\delta+1}\nabla\xi_1\|_{L^2(dx)}^2.
 \end{equation}
 And since $\delta>-1$, we see that 
 \begin{equation}\label{delt}
 (1+\abs{x})^{\delta+1}>1.
 \end{equation}
 When we combine the right hand side of (\ref{w0}) and (\ref{delt}), we see that $\nabla\xi_1\in L^2(dx)$.  Similarly, we find that $\nabla\zeta_1\in L^2(dx)$.  Therefore, we have obtained that $\xi_1,\zeta_1\in\mathscr{H}$.  We now use Lemma \ref{constlem} to conclude that $\xi-\xi_1$ and $\zeta-\zeta_1$ are constants.  From this, we conclude that $\xi$, and $\zeta$ approach constants at infinity.

 Moreover, since $u_0,v_0\rightarrow0$ as $\abs{x}\rightarrow\infty$, $e^{u_3}=\abs{x}^{\alpha_0}$, and $e^{v_3}=\abs{x}^{-\beta_0}$, we see that
\begin{equation}
\abs{\psi_{\uparrow}}^2,\abs{\psi_{\downarrow}}^2\rightarrow0\qquad\text{as}\qquad\abs{x}\rightarrow\infty.
\end{equation}
This confirms that the functions $\psi_{\uparrow}$ and $\psi_{\downarrow}$ do not satisfy the finite energy conditions imposed in \cite{medina,ichi2}.  Furthermore, the end behavior of these solutions classifies them as non-topological.
\end{proof}

\section{Results}\label{discuss}
In this paper, we studied the solutions to the system given by (\ref{main}) when the coupling matrix, $K$ is positive definite.  In \cite{medina}, the same system was studied and Medina was able to establish existence of topological solutions.  That is, these solutions were of finite energy.  In the present, we have established existence of non-topological solutions to (\ref{main}) and ultimately solutions to (\ref{vortex1}) that are of divergent energy.  What allowed us to consider such solutions is changing the function space from a standard Sobolev space to a weighted one by using a power weight on the $L^2$ norm.  Thus, we allowed for a different class of solutions.  These solutions would have an infinite norm in $W^{1,2}$, but would still be solutions to the system.  Under the appropriate power weight, we control the end behavior, and end up with solutions having a finite weighted norm.

Since the matrix $K$ is in terms of $p$ and $q$, and the filling factor, $\nu$ is in terms of $p$, we summarize the results of the paper in terms of the quantities $p$ and $q$ as well as the constraint parameters $\alpha$ and $\beta$ given by (\ref{constraints}). 

\begin{enumerate}
\item When $k_{12}<0$, the parameters of the matrix satisfy $q>p>0$.  In Theorem \ref{k12min}, we obtained the  necessary conditions for non-topological solutions, which decay exponentially to zero, to occur.  These necessary conditions are given by
\begin{equation}
0<\beta<\frac{\alpha}{4}\left(\frac{p}{q}+\frac{q}{p}-2\right).
\end{equation}

\item In Theorem \ref{51thm}, when $p>q>0$, the necessary conditions for non-topological solutions decaying to zero at infinity are given by
\begin{equation}
\alpha>0\qquad\text{and}\qquad \beta>\frac{\alpha}{4}\left(\frac{p}{q}+\frac{q}{p}-2\right).
\end{equation}

\item When $p>q>0$, and $\beta>\frac{\alpha}{4}\left(\frac{p}{q}+\frac{q}{p}-2\right)$, we obtained the necessary condition for saddle point type solutions in Theorem \ref{mountainthm}.  The necessary condition is given by $0<\alpha <\frac{32qp}{p^2-q^2}$.
\end{enumerate}

We now consider a special case, when the filling factor satisfies the much studied $\nu=5/2$ yielding $p=2\pi/5$ \cite{pfaf1,pfaf4}.  For this filling factor, we observe that there are two ranges for $q$.  When $k_{12}<0$, or $p<q$, we have $q>2\pi/5$ and then $\alpha$ and $\beta$ must satisfy 
\begin{equation}
0<\beta<\frac{\alpha}{4}\left(\frac{2\pi}{5q}+\frac{5 q}{2\pi}-2\right).
\end{equation}
As an example we take the pair $p=2\pi/5$ and $q={3\pi}/{5}$.  This yields the following conditions on $\alpha$ and $\beta$, $0<24\beta<\alpha,$ corresponding to the the following coupling matrix,
\begin{equation*}
K=\frac{1}{2}\begin{pmatrix}
5&-1\\-1&5
\end{pmatrix}.
\end{equation*}

If $0<k_{12}<1$ and $\nu=5/2$, then $0<q<2\pi/5$.  We observe that if we take $q=\pi/5$, we obtain the following restriction on the constraints $8\beta>\alpha>0$ corresponding to the following coupling matrix,
\begin{equation*}
K=\frac{1}{2}\begin{pmatrix}
3&1\\1&3
\end{pmatrix}.
\end{equation*}

An interesting question to ask would be if the coupling matrix were indefinite, that is $\abs{K}<0$, could such nontopological solutions be found?  That is, can one seek to establish existence of solutions in a weighted Sobolev space when $\abs{K}<0$.  In the present, we were unable to establish the coercivity of the energy functional for such $K$.  Another direction of further study would be for an arbitrary matrix, $K$.  In this scenario, the relationship between the system and the FQHE will be broken, but the problem will be interesting from a pure mathematical point of view.



\bibliography{infiniteenergy}

\begin{thebibliography}{10}

\bibitem{weighted3}
E.~A. Alarcon.
\newblock {Existence and finite dimensionality of the global attractor for a
  class of nonlinear dissipative equations}.
\newblock {\em Proc. R. Soc. Edinburgh Sect. A Math.}, 123(5):893--916, 1993.

\bibitem{pfaf2}
V.~M. Apalkov and T.~Chakraborty.
\newblock {Stable {P}faffian state in bilayer graphene}.
\newblock {\em Phys. Rev. Lett.}, 107(18):186803, 2011.

\bibitem{baer2014transport}
S.~Baer.
\newblock {\em {Transport Spectroscopy of Confined Fractional Quantum {H}all
  Systems}}.
\newblock PhD thesis, Springer, 2014.

\bibitem{nonconv}
A.~C. Balram.
\newblock {Interacting composite fermions: {N}ature of the 4/5, 5/7, 6/7, and
  6/17 fractional quantum {H}all states}.
\newblock {\em Phys. Rev. B}, 94(16):165303, oct 2016.

\bibitem{bt}
D.~Bartolucci and G.~Tarantello.
\newblock {Liouville type equations with singular data and their applications
  to periodic multivortices for the electroweak theory}.
\newblock {\em Commun. Math. Phys.}, 229(1):3--47, 2002.

\bibitem{bogo}
E.~B. Bogomol'nyi.
\newblock {The Stability of Classical Solutions}.
\newblock {\em Sov. J. Nucl. Phys.}, 24(4):449--454, 1976.

\bibitem{bondsling}
P.~Bonderson, A.~E. Feiguin, G.~M{\"{o}}ller, and J.~K. Slingerland.
\newblock {Competing Topological Orders in the $\nu$=12/5 Quantum {H}all
  State}.
\newblock {\em Phys. Rev. Lett.}, 108(3):36806, jan 2012.

\bibitem{chakraborty2013quantum}
T.~Chakraborty and P.~Pietil{\"{a}}inen.
\newblock {\em {The {Q}uantum {H}all {E}ffects}}, volume~85.
\newblock Springer Science {\&} Business Media, 2013.

\bibitem{fraser2016physics}
M.~D. Fraser, S.~H{\"{o}}fling, and Y.~Yamamoto.
\newblock {Physics and applications of exciton-polariton lasers}.
\newblock {\em Nat. Mater.}, 15(10):1049--1052, 2016.

\bibitem{frol1}
J.~Fr{\"{o}}hlich and P.~A. Marchetti.
\newblock {Quantum field theories of vortices and anyons}.
\newblock {\em Commun. Math. Phys.}, 121(2):177--223, 1989.

\bibitem{klitz}
A.~C. Gossard.
\newblock {The Quantized Hall Effect}.
\newblock {\em Sci. Am.}, 254(1-3):52--61, 1986.

\bibitem{griffiths2014principles}
J.~Griffiths.
\newblock {\em {Griffiths Testimonial Fund.}}, volume 166.
\newblock John Wiley {\&} Sons, 1905.

\bibitem{hans}
J.~Han and K.~Song.
\newblock {The existence and asymptotics of solutions for the Abelian
  {C}hern--{S}imons system with two {H}iggs fields and two gauge fields}.
\newblock {\em Nonlinear Anal. Theory, Methods Appl.}, 74(18):7426--7436, 2011.

\bibitem{hany}
X.~Han and Y.~Yang.
\newblock {Existence Theorems for Vortices in the
  {A}harony--{€"B}ergmanâ--{J}aferis--{M}aldacena Model}.
\newblock {\em Commun. Math. Phys.}, 333(1):229--259, 2014.

\bibitem{PhysRevB.95.125302}
J.~A. Hutasoit, A.~C. Balram, S.~Mukherjee, Y.-H. Wu, S.~S. Mandal, A.~Wojs,
  V.~Cheianov, and J.~K. Jain.
\newblock {The enigma of the $\nu=2+\frac{3}{8}$ fractional quantum {H}all
  effect}.
\newblock {\em Phys. Rev. B}, 95(12):125302, mar 2016.

\bibitem{ichi1}
I.~Ichinose and A.~Sekiguchi.
\newblock {Solitons in {C}hern--{S}imons theories of nonrelativistic $CP^{N-1}$
  models: {S}pin textures in the quantum {H}all effect}.
\newblock {\em Mod. Phys. Lett. A}, 12(30):19, 1997.

\bibitem{ichi2}
I.~Ichinose and A.~Sekiguchi.
\newblock {Topological solitons in {C}hern--{S}imons theories for the
  double-layer fractional quantum {H}all effect}.
\newblock {\em Nucl. Phys. B}, 493(3):683--706, 1997.

\bibitem{jackiw}
R.~Jackiw and S.~Y. Pi.
\newblock {Classical and quantal nonrelativistic {C}hern--{S}imons theory}.
\newblock {\em Phys. Rev. D}, 42(10):3500--3513, 1990.

\bibitem{jain}
J.~K. Jain.
\newblock {Composite-fermion approach for the fractional quantum {H}all
  effect}.
\newblock {\em Phys. Rev. Lett.}, 63(2):199--202, 1989.

\bibitem{laugh}
R.~B. Laughlin.
\newblock {Anomalous quantum {H}all effect: {A}n incompressible quantum fluid
  with fractionally charged excitations}.
\newblock {\em Phys. Rev. Lett.}, 50(18):1395--1398, may 1983.

\bibitem{levkivskyi2012mesoscopic}
I.~Levkivskyi.
\newblock {\em {Mesoscopic Quantum Hall Effect}}.
\newblock Springer Science {\&} Business Media, 2012.

\bibitem{lieby}
E.~H. Lieb and Y.~Yang.
\newblock {Non-Abelian Vortices in Supersymmetric Gauge Field Theory via Direct
  Methods}.
\newblock {\em Commun. Math. Phys.}, 313(2):445--478, 2012.

\bibitem{lin}
C.~S. Lin, G.~Tarantello, and Y.~Yang.
\newblock {Solutions to the master equations governing fractional vortices}.
\newblock {\em J. Differ. Equ.}, 254(3):1437--1463, 2013.

\bibitem{lopez1991fractional}
A.~Lopez and E.~Fradkin.
\newblock {Fractional quantum Hall effect and {C}hern--{S}imons gauge
  theories}.
\newblock {\em Phys. Rev. B}, 44(10):5246--5262, 1991.

\bibitem{mcowen1979behavior}
R.~C. McOwen.
\newblock {The behavior of the {L}aplacian on weighted {S}obolev spaces}.
\newblock {\em Commun. Pure Appl. Math.}, 32(6):783--795, 1979.

\bibitem{mcowen1984equation}
R.~C. McOwen.
\newblock {On the equation $\Delta u + Ke^{2u}= f$ and prescribed negative
  curvature in $\mathbb{R}^2$}.
\newblock {\em J. Math. Anal. Appl.}, 103(2):365--370, 1984.

\bibitem{medina}
L.~Medina.
\newblock {Vortex equations governing the fractional quantum {H}all effect}.
\newblock {\em J. Math. Phys.}, 56(9):91514, 2015.

\bibitem{weighted2}
A.~Mielke and G.~Schneider.
\newblock {Attractors for modulation equations on unbounded domains - existence
  and Comparison}.
\newblock {\em Nonlinearity}, 8(5):743--768, 1995.

\bibitem{nam}
K.~H. Nam.
\newblock {On the existence of self-dual vortices in the abelian
  {C}hern--{S}imons model with two {H}iggs fields}.
\newblock {\em J. Math. Anal. Appl.}, 406(1):101--110, 2013.

\bibitem{pfaf1}
M.~R. Peterson, T.~Jolicoeur, and S.~{Das Sarma}.
\newblock {Finite-layer thickness stabilizes the pfaffian state for the 5/2
  fractional quantum {H}all effect: Wave function overlap and topological
  degeneracy}.
\newblock {\em Phys. Rev. Lett.}, 101(1):16807, 2008.

\bibitem{jaffe1980vortices}
J.~Preskill.
\newblock {\em {Vortices and Monopoles}}.
\newblock Birkh{\"{a}}user, 1986.

\bibitem{pfaf4}
E.~Rezayi.
\newblock {Ground state of the 5/2 fractional quantum {H}all effect in the
  presence of {L}andau level mixing}.
\newblock {\em Bull. Am. Phys. Soc.}, 62, 2017.

\bibitem{taran}
O.~Savin.
\newblock {A {L}iouville theorem for solutions to the linearized
  {M}onge--{A}mpere equation}.
\newblock {\em Discret. Contin. Dyn. Syst.}, 28(3):865--873, 2010.

\bibitem{yangy2}
L.~Sibner, R.~Sibner, and Y.~Yang.
\newblock {Abelian gauge theory on {R}iemann surfaces and new topological
  invariants}.
\newblock In {\em Proc. R. Soc. A Math. Phys. Eng. Sci.}, volume 456, pages
  593--613. The Royal Society, 2000.

\bibitem{snizhko2012importance}
K.~Snizhko, V.~Cheianov, and S.~H. Simon.
\newblock {Importance of interband transitions for the fractional quantum
  {H}all effect in bilayer graphene}.
\newblock {\em Phys. Rev. B}, 85(20):201415, 2012.

\bibitem{sy2}
J.~Spruck and Y.~Yang.
\newblock {The existence of non-topological solitons in the self-dual
  {C}hern--{S}imons theory}.
\newblock {\em Commun. Math. Phys.}, 149(2):361--376, 1992.

\bibitem{sy3}
J.~Spruck and Y.~Yang.
\newblock {Topological solutions in the self-dual {C}hern--{S}imons theory:
  existence and approximation}.
\newblock In {\em Ann. l'Institut Henri Poincare Anal. Non Lineaire},
  volume~12, pages 75--97. Elsevier, 2016.

\bibitem{sy1}
J.~Spruck and Y.~S. Yang.
\newblock {On multivortices in the electroweak theory {II}. {E}xistence of
  {B}ogomol{'}nyi solutions in $\mathbb{R}^2$}.
\newblock {\em Comm. Math. Phys.}, 144(2):215--234, 1992.

\bibitem{pfaf3}
A.~Tylan-Tyler and Y.~Lyanda-Geller.
\newblock {In-plane electric fields and the $\nu=5/2$ fractional quantum {H}all
  effect in a disk geometry}.
\newblock {\em Phys. Rev. B}, 95(12):121302, 2017.

\bibitem{viefers2008quantum}
S.~Viefers.
\newblock {Quantum {H}all physics in rotating {B}ose-{E}instein condensates}.
\newblock {\em J. Phys. Condens. Matter}, 20(12):123202, 2008.

\bibitem{weighted1}
A.~Visintin.
\newblock {Chapter 8 Introduction to {S}tefan-{T}ype Problems}, 2008.

\bibitem{wangy}
S.~Wang and Y.~Yang.
\newblock {Abrikosov's vortices in the critical coupling}.
\newblock {\em SIAM J. Math. Anal.}, 23(5):1125--1140, 1992.

\bibitem{yangy1}
Y.~Yang.
\newblock {Topological solitons in the {W}einberg--{S}alam theory}.
\newblock {\em Phys. D Nonlinear Phenom.}, 101(1-2):55--94, 1997.

\bibitem{yangy3}
Y.~Yang.
\newblock {On a System of Nonlinear Elliptic Equations Arising in Theoretical
  Physics}.
\newblock {\em J. Funct. Anal.}, 170(1):1--36, 2000.

\bibitem{yang2013solitons}
Y.~Yang.
\newblock {\em {Solitons in field theory and nonlinear analysis}}.
\newblock Springer Science {\&} Business Media, 2001.

\end{thebibliography}
\bibliographystyle{abbrv}


%
%
%
\end{document}